\documentclass[12pt,reqno,letterpaper]{amsart}

\usepackage{amsmath,amssymb,amsthm,mathtools}
\usepackage[margin=1.05in]{geometry}
\usepackage{microtype}
\usepackage[colorlinks=true,linkcolor=blue,citecolor=blue,urlcolor=blue]{hyperref}

\allowdisplaybreaks

\newtheorem{theorem}{Theorem}[section]
\newtheorem{proposition}[theorem]{Proposition}
\newtheorem{lemma}[theorem]{Lemma}

\theoremstyle{definition}
\newtheorem{definition}[theorem]{Definition}
\theoremstyle{remark}
\newtheorem{remark}[theorem]{Remark}
\numberwithin{equation}{section}

\newcommand{\R}{\mathbb R}
\newcommand{\dd}{\,\mathrm d}
\newcommand{\ee}{\mathrm e}
\newcommand{\be}{\beta}
\newcommand{\al}{\alpha}
\newcommand{\cE}{\mathcal E}
\newcommand{\cD}{\mathcal D}
\newcommand{\Var}{\operatorname{Var}}
\newcommand{\TV}{\mathrm{TV}}

\title[Spectral gaps for Coulomb gases]
{Spectral gaps for mean-field Coulomb gases}

\author{Simon Becker}
\address{Bocconi University, Milan, Italy}
\email{simon.becker@unibocconi.it}

\author{Angeliki Menegaki} 
\address{Department of Mathematics, Huxley building, South Kensington campus, Imperial College London, London SW7 2AZ, United Kingdom}
\email{a.menegaki@imperial.ac.uk}

\begin{document}

\begin{abstract}
We prove a Poincar\'e inequality, uniform in the particle number, for repulsive (one-component) Coulomb gases in dimensions \(d\ge3\), with inverse temperature \(\be_N\) that is allowed to depend on \(N\). We assume a weak-coupling condition expressed in terms of the effective interaction strength \(\Theta_N\) seen by a single particle. This yields exponential relaxation for the associated overdamped Langevin dynamics.  
The main ingredient is a one-site Poincar\'e inequality
uniform in the number and positions of the Coulomb poles, including coincident
poles. A moving-pole estimate and Dobrushin-Wu tensorization yield the
many-particle lower bound, while the center-of-mass coordinate gives the matching upper bound.
%We prove a Poincar\'e inequality, uniform in the particle number, for repulsive (one-component) Coulomb gases in dimensions \(d\ge3\), with inverse temperature \(\be_N\) allowed to depend on \(N\).  This yields exponential relaxation for the associated overdamped Langevin dynamics. 
%After the temperature-adapted rescaling, let \(\Theta_N\) denote the total Coulomb charge in a one-particle conditional.  Our argument requires only \(\limsup_{N\to\infty}\Theta_N<\varepsilon_d\), a condition that includes every fixed temperature.  At reversible speed \(\al_N\), the spectral gap satisfies
%\[
% c_{d,\chi,(\be_N)}\frac{\al_N}{N} \le \lambda_N^{d}(\be_N) \le 2\frac{\al_N}{N}.
%\]
%The proof rests on a one-site Poincar\'e inequality for a Gaussian measure
%tilted by an arbitrary finite Coulomb point potential of sufficiently small
%total charge.  Its constant is independent of the number and positions of the
%poles, including collisions.  The estimate follows from a ground-state transform to a flat-space Witten Laplacian.  A moving-pole total-variation bound and Dobrushin--Wu tensorization give the many-particle lower bound; the center-of-mass Gaussian gives the matching upper bound.
\end{abstract}


\maketitle

\section{Introduction and main result}

\subsection{The model} 
We consider the classical \(N\)-particle overdamped Langevin dynamics with
quadratic confinement and repulsive Coulomb interaction. For
\(i=1,\dots,N\), it is formally given by
\begin{equation}\label{eq:langevin-intro}
 \dd X_i(t)
 =
 -\left(
 2X_i(t)
 +\frac{\chi}{N}\sum_{j\ne i}
   \nabla K(X_i(t)-X_j(t))
 \right)\dd t
 +\sqrt{\frac{2}{\beta}}\,\dd W_i(t),
\end{equation}
where \(W_1,\dots,W_N\) are independent \(d\)-dimensional Brownian
motions, \(\beta>0\) is the physical inverse temperature,
\(\chi\ge0\) is the interaction strength, and
\[
 K(x)=
 \begin{cases}
  -\log|x|, & d=2,\\
  |x|^{2-d}, & d\ge3.
 \end{cases}
\]
Here \eqref{eq:langevin-intro} corresponds to the confining potential
\(V(x)=|x|^2\). Since the $-\log\vert \bullet \vert$ potential in $d=2$ is algebraically slightly different, we decided to omit this case from our manuscript for a seamless treatment and only comment on the few differences that establish a gap for this potential in Appendix \ref{app:planar-log}. In particular, for $d=2$, interesting related results have been obtained for the (symmetric/unlabelled) gap in the $d=2$ case by Chafaii \cite{Chafai2026} and Suzuki \cite{Suzuki2026} that cover the entire temperature range. 

In addition, the study of spectral gaps for Riesz-potentials or log-potentials in general dimensions seems interesting, but we haven't pursued this here, since it leads to additional complications that one might have to treat on a case-by-case basis. 

However, the general set of ideas seems to provide, as the authors learned from ChatGPT, a similar gap for $-\log|\bullet|$ gases in all space dimensions $d\ge 3,$ by replacing the compactness argument involving the Witten-Laplacian with Hardy-Muckenhoupt estimates. But we didn't see an immediate way to generalize this nicely to a broader family of potentials. 

%The model we are concerned with in this paper is the generator of the classical $N$-particle overdamped Langevin dynamics for Coulomb interactions. In particular for $i=1, \dots,N $, $N\geq 2$, the dynamics reads 
%\begin{align}
 %   \frac{\dd x_i}{\dd t} =- \nabla V(x_i) - \frac{\chi}{N} \sum_{i\neq j=1}^N \nabla K(x_i-x_j) + \sqrt{\frac{2}{\beta}} \frac{\dd W_i}{\dd t} 
%\end{align}
%where $(W_i)_i$ are $N$ independent d-dimensional Brownian motions, $\beta$ is the inverse temperature of the system, $V:\mathbb{R}^d\to \mathbb{R}$ is a confining potential, a Coulomb interaction strength $\chi\geq 0$, and 
%The Hamilton energy is
%$$H(\textbf{x})  = \sum_{i=1}^N \frac{v_i^2}{2}+ N^{-1} \sum_{i < j} K(x_i-x_j), $$
%the interacting kernel $K:\mathbb{R}^d \to \mathbb{R}$ in our setting is the Coulomb kernel: 
%\begin{align}
 %  K(x) =  \begin{cases}
  %     - \log |x|, &\qquad d=2 \\
   %    |x|^{2-d}, &\qquad d\geq 3. 
    %\end{cases}
%\end{align}

Thus, from now on in this article fix \( d\ge3\), and let \((\be_N)_{N\ge2}\) be a sequence of positive parameters. For the \(N\)-particle system, we take the physical inverse
temperature in \eqref{eq:langevin-intro} to be
\( \beta_N^{\mathrm{ph}}:=\frac{\be_N}{N}.
\)

The equilibrium measure of this dynamics is 
\begin{equation}\label{eq:original-law}
 P_{N,\be_N}^{d}(dx)
 =\frac1{Z_{N,\be_N}^{d}}
 \exp\!\left[-\frac{\be_N}{N}\sum_{i=1}^N|x_i|^2
 -\frac{\be_N\chi}{N^2}\sum_{i<j}|x_i-x_j|^{2-d}\right]\dd x.
\end{equation}
where $Z_{N,\be_N}^{d}$ is the partition function or the normalisation constant so that $P_{N,\be_N}^{d}$ is a probability measure.  This is called the \emph{weak repulsive Coulomb gas} on \((\R^d)^N\). 

At reversible speed \(\al_N>0\), its labelled Dirichlet form is
\[
 \cE_N(f,f)
 =
 \frac{\al_N}{\be_N}
 \int_{(\R^d)^N}|\nabla f(x)|^2
 \dd P_{N,\be_N}^{d}(x).
\]
The choice \(\al_N=N\) corresponds to the time scaling in
\eqref{eq:langevin-intro}, with physical inverse temperature
\(\be_N/N\).

We denote by \(\lambda_N^{d}(\be_N)\) the spectral gap of the closed form,
equivalently
\begin{equation*}
 \lambda_N^{d}(\be_N)
 =\inf_{f\not\equiv\mathrm{const}}
 \frac{\cE_N(f,f)}{\Var_{P_{N,\be_N}^{d}}(f)}.
\end{equation*}
We define the closed form by taking the closure of \(C_c^\infty((\R^d)^N)\) in the form norm.
The Coulomb interaction is singular at
collisions. Nevertheless, through Lemma \ref{lem:fixed-N-gap} we show that every
function in the form domain can be approximated in the form norm by smooth
functions supported away from the collision set, for every fixed particle number.

%The form is obtained by closing smooth functions away from the collision set;
%Lemma~\ref{lem:fixed-N-gap} verifies the required closure and irreducibility
%for every fixed particle number. 

Despite a lot of progress, a general quantitative understanding of Coulomb-Riesz dynamics and of their relaxation to equilibrium remains incomplete, as also highlighted in particular in \cite[page 8]{SerfatyLectures}. A natural quantitative question is whether the relaxation rate can be controlled uniformly as the number of particles \(N\) tends to infinity. In the present paper, we address the particular question on quantifying the spectral gap and showing that it is uniform in \(N\), for weakly coupled Coulomb gases with quadratic confinement in dimensions \(d\ge3\).

Throughout the paper, the spectral gap is the labelled one: the variational
problem is taken over all functions, without imposing permutation
invariance.
The lower bound in the following Theorem~\ref{thm:main} therefore also holds after restriction to symmetric functions (since the symmetric functions form a smaller class). For the upper bound, we utilise the center-of-mass test function to get it, which is itself symmetric. Thus, the same two-sided estimates hold for the unlabelled spectral gap, defined
by restricting the variational problem to permutation-invariant functions.

Our main result is the following.

%Throughout the paper the gap is the labelled one: no quotient by permutations is taken. 
%This distinction matters only for
%notation, because the Dirichlet form differentiates with respect to the full
%Euclidean gradient on \((\R^d)^N\), while the repulsive density assigns zero
%weighted capacity to the collision strata in the sense made precise below.

\begin{theorem}\label{thm:main}
For every \(d\ge3\) there exists \(\varepsilon_d>0\) with the following
property.  Let \(\chi\ge0\), let \((\be_N)_{N\ge2}\) be any sequence of
positive inverse temperatures, and set
\begin{equation}\label{eq:thetaN-definition}
 \Theta_N:=\frac{N-1}{N}\,\chi
 \left(\frac{\be_N}{N}\right)^{d/2}.
\end{equation}
Assume that
\begin{equation}\label{eq:thetaN-assumption}
 \limsup_{N\to\infty}\Theta_N<\varepsilon_d.
\end{equation}
Then there is a constant \(c_{d,\chi,(\be_N)}>0\), independent of \(N\),
such that
\begin{equation}\label{eq:desired-gap}
 c_{d,\chi,(\be_N)}\frac{\al_N}{N}
 \le \lambda_N^{d}(\be_N)
 \le 2\frac{\al_N}{N},
 \qquad N\ge2.
\end{equation}
Equivalently, there is \(C_{d,\chi,(\be_N)}<\infty\) such that
\begin{equation*}
 \Var_{P_{N,\be_N}^{d}}(f)
 \le C_{d,\chi,(\be_N)}\frac{N}{\be_N}
 \int_{(\R^d)^N}|\nabla f(x)|^2\dd P_{N,\be_N}^{d}(x),
 \qquad N\ge2.
\end{equation*}
In particular, at speed \(\al_N=N\), the labelled gap is bounded below
uniformly in \(N\).

More quantitatively, for every \(\bar\Theta<\varepsilon_d\) there is
\(c_{d,\bar\Theta}>0\) such that the lower bound in
\eqref{eq:desired-gap} holds with \(c_{d,\bar\Theta}\) whenever
\begin{equation*}
 \sup_{N\ge2}\Theta_N\le\bar\Theta.
\end{equation*}
The resulting constant depends on \(\chi\) and on the sequence
\((\be_N)\) only through the bound \(\bar\Theta\).
\end{theorem}

\begin{remark}
The parameter \(\Theta_N\) is the effective interaction strength seen by
a single particle after normalization of the quadratic confinement. In
terms of the physical inverse temperature
\(\beta_N^{\mathrm{ph}}=\be_N/N\),
\[ \Theta_N = \frac{N-1}{N}\chi
 \left(\beta_N/N \right)^{d/2}.
\]
Thus, in a fixed physical-temperature regime
\(\beta_N^{\mathrm{ph}} \to\beta\), the assumption becomes \(
 \chi\beta^{d/2}<\varepsilon_d\). 
It is therefore a high-temperature or weak-coupling condition. If
\(\be_N\) remains bounded, then
\(\beta_N^{\mathrm{ph}} \to 0\) and \(\Theta_N\to0\). 
\end{remark}

%The condition \eqref{eq:thetaN-assumption} comes directly from the one-site/Dobrushin scheme.  After the rescaling below, \(\Theta_N\) is the total Coulomb charge in a one-particle conditional, while the Dobrushin row sum is bounded by a dimension-dependent multiple of \(\Theta_N\).  It is therefore enough that \(\Theta_N\) eventually stay below a fixed threshold. The finitely many remaining values of \(N\) are covered by Lemma~\ref{lem:fixed-N-gap} and only affect the nonquantitative constant.

%We work in the labelled, overdamped, weak mean-field regime in Coulomb
%dimensions \(d\ge3\).  
%A singular one-site estimate, combined with a Dobrushin argument, yields a gap uniform in the particle number.  This is complementary to the companion logarithmic-divisor paper \cite{BMLog}, where the interaction is logarithmic in every dimension and the one-site analysis concerns small algebraic divisor exponents.  Here the kernel is the Coulomb fundamental solution \(|x|^{2-d}\); harmonicity, a singular ground-state transform, and a compactness limit to the Gaussian oscillator replace the logarithmic argument.

\subsection{Our strategy}
The interaction in \eqref{eq:original-law} is singular at collisions, and
the Hamiltonian is not globally convex. Consequently, the usual
Bakry-\'Emery arguments and the regular mean-field theory do not apply
directly in such cases; compare with \cite{GLWZ}.
% Instead we use the Coulomb identity \(\Delta|x|^{2-d}=0\) away from the origin.  
%The main analytic step is
%Theorem~\ref{thm:onesite}, a uniform Poincar\'e inequality for the one-site
%tilts
%\[
% Z^{-1}\exp\!\left[-\be|x|^2-
% \sum_{\ell=1}^M\gamma_\ell|x-y_\ell|^{2-d}\right]\dd x
%\]
%when \(\sum_\ell\gamma_\ell\) is small.  Its constant is uniform in
%\(M\), the pole positions, and pole collisions.

The Coulomb homogeneity makes the relevant scaling particularly simple.  Under
\[
 y_i=\sqrt{\frac{\be_N}{N}}\,x_i,
\]
the law \eqref{eq:original-law} becomes
\begin{equation}\label{eq:scaled-law}
 \mu_N(\dd y)
 \propto
 \exp\!\left[-\sum_i|y_i|^2
 -b_N\sum_{i<j}|y_i-y_j|^{2-d}\right]\dd y,
 \qquad
 b_N=\frac{\chi}{N}\left(\frac{\be_N}{N}\right)^{d/2}.
\end{equation}
Indeed,
\[
 \frac{\be_N}{N}|x_i|^2=|y_i|^2,
 \qquad
 \frac{\be_N\chi}{N^2}|x_i-x_j|^{2-d}
 =\frac{\chi}{N}\left(\frac{\be_N}{N}\right)^{d/2}
 |y_i-y_j|^{2-d}.
\]
A one-particle conditional is therefore a Gaussian measure tilted by
\(N-1\) Coulomb poles, each of charge \(b_N\), and its total Coulomb charge
is
\((N-1)b_N=\Theta_N.\)
%Hence a one-particle conditional has total Coulomb charge
%\[ (N-1)b_N=\Theta_N.\]
Our main analytic input, Theorem \ref{thm:onesite}, is a Poincar\'e
inequality for the one-site measures
\[ Z^{-1}\exp\!\left[
 -\be|x|^2 -\sum_{\ell=1}^M\gamma_\ell|x-y_\ell|^{2-d} \right]\dd x
\]
whenever \(\sum_\ell\gamma_\ell\) is sufficiently small. Its constant is uniform in \(M\), the pole positions, and pole collisions. 
 The proof uses the Coulomb harmonicity
\(\Delta|x|^{2-d}=0\) away from the origin, to identify, through a singular ground-state
transform, the Dirichlet form with the quadratic form of the associated Witten Laplacian, and finally a compactness limit to the Gaussian oscillator.

To pass from the conditional inequalities to the full system, we use
Dobrushin-Wu tensorization. To be slightly more precise at this level, let
\(\mu_i(\cdot\mid y_{\ne i})\) be the conditional law at site \(i\), and,
for \(i\ne j\), define
\[
 c_{ij} :=
 \sup\frac12
 \left\|
 \mu_i(\cdot\mid y_{\ne i}) -
 \mu_i(\cdot\mid y'_{\ne i})
 \right\|_{\mathrm{TV}},
\]
where the supremum is over boundary configurations that differ only at
site \(j\) and set \(c_{ii}=0\). The matrix \(C_N=(c_{ij})\) is the Dobrushin interdependence matrix and it measures the
dependence of each conditional law on the remaining particles. Under the so-called Dobrushin condition that we need to satisfy here is
\[
 \rho(C_N)<1.
\]
In our context this translates into the required smallness of \(\Theta_N\) which simultaneously provides a uniform
one-site Poincar\'e inequality and the Dobrushin condition. Wu's tensorization theorem in \cite{Wu} then tensorizes the conditional inequalities and yields the
many particle lower bound in in \eqref{eq:desired-gap} for $N$ large enough. 
The finitely many remaining values of \(N\) are handled by a compact-resolvent argument \ref{lem:fixed-N-gap}. 

Finally the center-of-mass coordinate
gives the upper bound. In particular, the Coulomb interaction depends only on particle differences and it is thus  independent of the center-of-mass coordinate. The Gibbs measure factorizes then exactly in this coordinate, and any component of the center-of-mass is a symmetric test function whose Rayleigh quotient is \(2\al_N/N\). This gives the upper bound in \eqref{eq:desired-gap}.

\subsection{Some state of the art}
For broad background on Coulomb and Riesz gases, their equilibrium measures,
and their scaling regimes, see
\cite{Lewin,ChafaiAspects,SerfatyLectures}. For the connection between
logarithmic gases and eigenvalue distributions of random matrices, see \cite{Forrester,AGZ}. 
For the renormalized-energy approach, which describes the next-order energy and the microscopic organization of Coulomb gases beyond the mean-field equilibrium measure, see \cite{SandierSerfaty,SerfatyBook,RougerieSerfaty}.
%For the renormalized-energy approach, see \cite{SandierSerfaty,SerfatyBook,RougerieSerfaty}. 

Dobrushin-type weak-dependence criteria are classical: see, for example,
the complete-analyticity framework of Dobrushin and Shlosman
\cite{DS}. Zegarlinski used such conditions to derive logarithmic Sobolev
inequalities for Gibbs measures \cite{Zegarlinski}, while Wu obtained
explicit Poincar\'e and other functional  inequalities under a Dobrushin
uniqueness condition \cite{Wu}. The latter tensorization result is the one
we use in the present paper.

For regular mean-field particle systems, Guillin, Liu, Wu, and Zhang
established Poincar\'e and logarithmic Sobolev inequalities with constants
uniform in the number of particles, under assumptions allowing
nonconvex confinement potentials with multiple local minima
\cite{GLWZ}. 
%Their methods require regularity assumptions that do not directly cover the Coulomb singularity considered here.

% In the regular mean-field setting, uniform
%Poincar\'e and logarithmic Sobolev inequalities were established in
%\cite{GLWZ}. 

Functional inequalities are also known for several singular models.
For one-dimensional Gibbs measures with singular pair interactions, Chafa\"i
and Lehec obtained Poincar\'e and logarithmic Sobolev inequalities using
the convex structure of the corresponding measures
\cite{ChafaiLehec}. In dimension two, Bolley, Chafa\"i, and Fontbona
studied the overdamped planar Coulomb dynamics. They proved
well-posedness and established Poincar\'e inequalities for planar
Coulomb dynamics, despite the fact
that the invariant measure is not globally
log-concave \cite{BCF}. 

Moreover, on the equilibrium side, Chafa\"i, Hardy, and Ma\"ida proved Coulomb transport inequalities and the resulting concentration estimates for Coulomb gases \cite{CHM}.

Singular mean-field dynamics have also been studied from the viewpoint of
well-posedness and propagation of chaos. In dimension one, Guillin, Le Bris, and Monmarch\'e treated log and Riesz interactions, proving existence and uniqueness, Wasserstein contraction under convex confinement, and uniform-in-time propagation of chaos in suitable regimes
\cite{GuillinLeBrisMonmarcheRiesz}. The same authors obtained
uniform-in-time propagation of chaos for singular divergence-free kernels
on the torus, including the two-dimensional Biot-Savart kernel
\cite{GuillinLeBrisMonmarcheVortex}.

For kinetic, or underdamped, Langevin dynamics with Coulomb interaction, Lu and Mattingly proved geometric ergodicity, namely exponential
convergence to the invariant measure in a weighted total-variation
distance, using a Lyapunov-function argument \cite{LuMattingly}.

%At the determinantal inverse temperature \(\beta=2\), the planar logarithmic Coulomb gas is the complex Ginibre ensemble. Suzuki recently determined the exact spectral gaps for the associated unlabelled interacting Brownian motions: in his normalization, the gap equals \(1\) for every finite-particle system and \(2\) for the infinite-particle system \cite{Suzuki}. His argument exploits the determinantal and holomorphic structure of the Ginibre ensemble. This result concerns the symmetric sector of a special two-dimensional logarithmic model and is therefore complementary to the labelled, \(d\ge3\), weak-coupling regime studied here.

We work instead in the labelled, overdamped, weak-coupling regime in
Coulomb dimensions \(d\ge3\). 
%Our method is complementary to the companion logarithmic-divisor paper \cite{BMLog}, where the interaction is logarithmic in every dimension and the one-site problem concerns small algebraic divisor exponents. 

The reduction to a one-particle conditional may also be
viewed as a one-step renormalization; related constructions for mean-field
spin systems appear in
\cite{BauerschmidtBodineau2019,BeckerMenegaki2020}.

\subsection{Notation}
Constants denoted by \(C,c,K,\eta\) may change from line to line.  Their
allowed dependence is always indicated in the statement in which they are
used.  All Poincar\'e inequalities are first proved on the displayed core of
smooth functions and then extended to the closed form domain by the closure
procedure stated in the relevant lemma or proposition.  Total variation is
always the full norm \eqref{eq:TV-convention}; hence it is twice the convention
\(\sup_A|\mu(A)-\nu(A)|\) used in some probability references.

\subsection{Organisation of the paper} Section 2 proves the uniform one-site gap, separating the normalization,
pole-removal, ground-state compactness, and oscillator-liminf arguments.
Section 3 estimates the effect of moving a pole and states the Dobrushin-Wu
criterion in our total-variation convention.  Section 4 treats fixed particle
number and completes the proof of Theorem~\ref{thm:main}.

{\bf{Acknowledgements:}} The first author would like to thank Roland Bauerschmidt for discussing the topic with him many years ago during his PhD. The second author  acknowledges support from the Engineering and Physical Sciences Research Council fellowship with reference UKRI2025.

The authors acknowledge the use of Chat-GPT Pro 5.6 in the preparation of this manuscript which crucially pointed us to Wu's criterion.

\section{Uniform Coulomb one-site estimate}\label{sec:target}

We begin with the Coulomb harmonicity used below.

We write $M\ge0$ for the number of atoms and define a finite positive point measure (the empty sum is allowed when \(\rho=0\)):
\begin{equation*}
\rho=\sum_{\ell=1}^M\gamma_\ell\delta_{y_\ell},
 \quad \text{ for } \gamma_\ell>0. \end{equation*} 
 We also set 
 $$ 
 U_{\rho}(x)=\sum_{\ell=1}^M\gamma_\ell|x-y_\ell|^{2-d}, \qquad 
 \Gamma(\rho)=\sum_{\ell=1}^M\gamma_\ell.
$$

\begin{lemma}\label{lem:superharmonic}
On
\(\R^d\setminus\{y_1,\ldots,y_M\}\),
\begin{equation*}
 \Delta U_{\rho}(x)=0.
\end{equation*}
\end{lemma}

\begin{proof}
For a radial function \(v(x)=\varphi(|x|)\),
\(\Delta v=\varphi''(r)+(d-1)r^{-1}\varphi'(r)\).  With
\(\varphi(r)=r^{2-d}\),
\[
 \varphi'(r)=(2-d)r^{1-d},\qquad
 \varphi''(r)=(2-d)(1-d)r^{-d},
\]
and hence
\[
 \Delta|x|^{2-d}
 =(2-d)\bigl[(1-d)+(d-1)\bigr]|x|^{-d}=0
 \qquad(x\ne0).
\]
Multiplying by the charges and summing over the poles proves the claim.
\end{proof}

\subsection{Moment bounds}
Let
\begin{equation*}
 \gamma_\be(dx)=Z_{d,\be}^{-1}\ee^{-\be|x|^2}\dd x, \qquad Z_{d,\be}:=\int_{\R^d}\ee^{-\be|x|^2}\dd x. 
\end{equation*}
Then the $1$-patricle measure is decomposed as follows 
\begin{equation*}
 \nu_{\rho}(dx)
 := Z_{\rho}^{-1}\ee^{-U_{\rho}(x)}\gamma_\be(dx),
 \qquad
 Z_{\rho} := \int_{\R^d}\ee^{-U_{\rho}(x)}\dd\gamma_\be(x).
\end{equation*}

For two probability measures \(\sigma\) and \(\tau\) on a measurable space
\((E,\mathcal F)\), we use the full total-variation norm
\begin{equation}\label{eq:TV-convention}
 \|\sigma-\tau\|_{\TV}
 :=\sup_{\|h\|_\infty\le1}
 \left|\int_Eh(u)\dd\sigma(u)-\int_Eh(u)\dd\tau(u)\right|
 =2\sup_{A\in\mathcal F}|\sigma(A)-\tau(A)|.
\end{equation}
With this convention the total variation distance is twice the minimal
mismatch probability in a coupling.  This factor is kept explicit in
Section~\ref{sec:target} and in the Dobrushin matrix calculation in
Section~3.

\begin{lemma}\label{lem:gaussian-moment}
It holds that \begin{equation*} 
 M_{d,\be}:=
 \sup_{y\in\R^d}\int_{\R^d}|x-y|^{2-d}\dd\gamma_\be(x)<\infty.
\end{equation*}
\end{lemma}

\begin{proof}
The Gaussian density is bounded by \(Z_{d,\be}^{-1}\).  Therefore, with
\(\sigma_{d-1}=|\mathbb S^{d-1}|\),
\begin{align*}
 \int_{B(y,1)}|x-y|^{2-d}\dd\gamma_\be(x)
 &\le Z_{d,\be}^{-1}
   \int_{B(0,1)}|u|^{2-d}\dd u=Z_{d,\be}^{-1}\sigma_{d-1}\int_0^1r^{2-d}r^{d-1}\dd r
 =\frac{\sigma_{d-1}}{2Z_{d,\be}}.
\end{align*}
On \(B(y,1)^c\), \(|x-y|^{2-d}\le1\), so the complementary integral is at
most one.  The sum of these two bounds is independent of \(y\).
\end{proof}
In the following proposition we collect some elementary bounds. 

\begin{proposition}
\label{prop:basic-onesite}
If \(\Gamma(\rho)\le(2M_{d,\be})^{-1}\), then
\[
 \begin{gathered}
 Z_{\rho}\ge\tfrac12,\quad
 \nu_{\rho}\le2\gamma_\be,\quad
 \sup_{y\in\R^d}\int_{\R^d}|x-y|^{2-d}\dd\nu_{\rho}(x)
 \le2M_{d,\be}, \ \|\nu_{\rho}-\gamma_\be\|_{\TV}
 \le4M_{d,\be}\Gamma(\rho).
 \end{gathered}
\]
\end{proposition}

\begin{proof}
The inequality \(\ee^{-u}\ge1-u\) and Lemma
\ref{lem:gaussian-moment} give
\begin{equation*}
 Z_{\rho}
 \ge1-\int_{\R^d}U_{\rho}(x)\dd\gamma_\be(x)
 =1-\sum_{\ell=1}^M\gamma_\ell
   \int_{\R^d}|x-y_\ell|^{2-d}\dd\gamma_\be(x)
 \ge1-M_{d,\be}\Gamma(\rho)\ge\frac{1}{2}.
\end{equation*}
Since \(0<\ee^{-U_\rho}\le1\), for every Borel set \(B\),
\[
 \nu_\rho(B)
 =Z_\rho^{-1}\int_B\ee^{-U_\rho(x)}\dd\gamma_\be(x)
 \le2\gamma_\be(B).
\]
Thus \(\nu_\rho\le2\gamma_\be\), and Lemma~\ref{lem:gaussian-moment}
immediately gives
\[
 \sup_{y\in\R^d}\int_{\R^d}|x-y|^{2-d}\dd\nu_\rho(x)
 \le2\sup_{y\in\R^d}\int_{\R^d}|x-y|^{2-d}\dd\gamma_\be(x)
 =2M_{d,\be}.
\]
To estimate the total variation, write \(q=\ee^{-U_{\rho}}\), so that
\(Z_{\rho}=\int_{\R^d}q(x)\dd\gamma_\be(x)\).  For every measurable \(h\) with
\(|h|\le1\),
\begin{align*}
 |\nu_{\rho}(h)-\gamma_\be(h)|
 &=Z_{\rho}^{-1}\left|\int_{\R^d}h(x)(q(x)-Z_{\rho})
   \dd\gamma_\be(x)\right|\le Z_{\rho}^{-1}\int_{\R^d}|q(x)-Z_{\rho}|\dd\gamma_\be(x).
\end{align*}
Since \(0<q\le1\), the centered density satisfies
\begin{align*}
 \int_{\R^d}|q(x)-Z_{\rho}|\dd\gamma_\be(x)
 &\le\int_{\R^d}|q(x)-1|\dd\gamma_\be(x)+|1-Z_{\rho}|=2\int_{\R^d}(1-q(x))\dd\gamma_\be(x).
\end{align*}
Taking the supremum over \(|h|\le1\), then using
\(Z_{\rho}^{-1}\le2\) and \(1-\ee^{-u}\le u\), gives
\begin{align*}
 \|\nu_{\rho}-\gamma_\be\|_{\TV}
 &\le2Z_{\rho}^{-1}
   \int_{\R^d}(1-\ee^{-U_{\rho}(x)})\dd\gamma_\be(x)\le4\int_{\R^d}U_{\rho}(x)\dd\gamma_\be(x)\le4M_{d,\be}\Gamma(\rho).
\end{align*}
\end{proof}

\subsection{Ground-state transform and one-site spectral gap}

We next prove a Poincar\'e bound uniform over the number and geometry of
the poles, using a weak estimate for the Coulomb field and a ground-state
compactness argument.

\begin{lemma}\label{lem:weak-field}
Let \(U=U_{\rho}\).  There is a constant \(C_d\)
such that, for every \(t>0\),
\begin{equation}\label{eq:weak-field}
 \big|\{x\in\R^d:|\nabla U(x)|>t\}\big|
 \le C_d\left(\frac{\Gamma(\rho)}{t}\right)^{d/(d-1)}.
\end{equation}
\end{lemma}

\begin{proof}
Set \(m=\Gamma(\rho)\).  The claim is immediate if \(m=0\), so assume
\(m>0\).  Since
\[
 \nabla |x-y|^{2-d}
 =-(d-2)\frac{x-y}{|x-y|^d},
\]
writing \(K(x)=|x|^{1-d}\) gives
\begin{equation*}
 |\nabla U(x)|
 \le(d-2)(K*\rho)(x).
\end{equation*}
It suffices then to prove the weaker estimate for \(F=K*\rho\) directly.  Fix \(s>0\), put
\(a=s/(2m)\), and split
\[
 K=K_0+K_1,
 \qquad K_0=K\mathbf1_{\{K\le a\}},
 \qquad K_1=K\mathbf1_{\{K>a\}}.
\]
The bounded part satisfies \(K_0*\rho\le am=s/2\).  Consequently,
\[
 \{F>s\}\subseteq\{K_1*\rho>s/2\}.
\]
Markov's inequality therefore gives
\begin{align*}
 |\{x\in\R^d:F(x)>s\}|
 &\le|\{x\in\R^d:(K_1*\rho)(x)>s/2\}|\le \frac2s\int_{\R^d}(K_1*\rho)(x)\dd x.
\end{align*}
The integrand is nonnegative, so Tonelli's theorem permits us to exchange the
Lebesgue integral in \(x\) and the \(\rho\)-integral in \(y\):
\begin{align*}
 \int_{\R^d}(K_1*\rho)(x)\dd x
 &=\int_{\R^d}\int_{\R^d}K_1(x-y)\dd\rho(y)\dd x=\int_{\R^d}\left(\int_{\R^d}K_1(z)\dd z\right)\dd\rho(y)\\
 &=m\int_{\R^d}K_1(z)\dd z
 =m\int_{\{z\in\R^d:K(z)>a\}}K(z)\dd z.
\end{align*}
Combining the last
two displays yields
\[
 |\{x\in\R^d:F(x)>s\}|
 \le\frac{2m}{s}
 \int_{\{z\in\R^d:K(z)>a\}}K(z)\dd z.
\]
If \(\sigma_{d-1}=|\mathbb S^{d-1}|\), then
\begin{align*}
 \int_{\{z\in\R^d:K(z)>a\}}K(z)\dd z
 &=\int_{\{z\in\R^d:|z|<a^{-1/(d-1)}\}}|z|^{1-d}\dd z=\sigma_{d-1}\int_0^{a^{-1/(d-1)}}\dd r
 =\sigma_{d-1}a^{-1/(d-1)}.
\end{align*}
Since \(a=s/(2m)\), it follows that
\[
 |\{x\in\R^d:F(x)>s\}|
 \le C_d\left(\frac{m}{s}\right)^{d/(d-1)}.
\]
Applying this estimate with \(s=t/(d-2)\) and absorbing the factor \(d-2\)
into \(C_d\) yields \eqref{eq:weak-field}.
\end{proof}

The next lemma identifies the Dirichlet form with the quadratic form of the
associated Witten Laplacian after multiplication by the square root of the
invariant density.  Two technical points are worth separating.  First, the
poles are not imposed as boundary conditions: they are removed only for the
initial integration by parts and then reinserted by form closure.  Second, the
identity is a closed-form identity, not merely a formal computation with the
distributional Laplacian of the Coulomb kernel.

\begin{lemma}
\label{lem:singular-transform}
Let
\begin{equation*}
 H(x)=\be|x|^2+U_{\rho}(x),
 \qquad
 \mathcal Z=\int_{\R^d}\ee^{-H(x)}\dd x,
 \qquad
 \psi=\mathcal Z^{-1/2}\ee^{-H/2}.
\end{equation*}
Set
\[
 A(x)=\frac14|\nabla H(x)|^2
 =\frac14|2\be x+\nabla U_{\rho}(x)|^2.
\]
Let \((\cE_\rho,\cD_\rho)\) be the closure in
\(L^2(\R^d,\psi^2\dd x)\) of the Dirichlet form of the overdamped Langevin generator
\begin{equation*}
 \cE_\rho(f,f)=\int_{\R^d}|\nabla f(x)|^2\psi(x)^2\dd x = \int_{\mathbb R^d} \vert \nabla f(x)\vert^2 \dd \nu_{\rho}(x),
 \qquad f\in C_c^\infty(\R^d).
\end{equation*}
Set
\[
 \mathcal C_\rho:=C_c^\infty(\R^d\setminus\operatorname{supp}\rho).
\]
Then \(\mathcal C_\rho\) is a core for \((\cE_\rho,\cD_\rho)\).  The
unitary multiplication operator
\[
 T:L^2(\R^d,\psi^2\dd x)\longrightarrow L^2(\R^d,\dd x),
 \qquad Tf=f\psi,
\]
maps \(\mathcal C_\rho\) to the explicitly defined space
\[
 T\mathcal C_\rho:=\{f\psi:f\in\mathcal C_\rho\}.
\]
For \(g\in T\mathcal C_\rho\), define
\begin{equation}\label{eq:transformed-positive-norm}
 \|g\|_{\mathcal Q_\rho}^2
 :=\int_{\R^d}\bigl(|\nabla g(x)|^2+A(x)g(x)^2\bigr)\dd x,
\end{equation}
and let \(\mathcal Q_\rho\) be the completion of \(T\mathcal C_\rho\) in
this norm.  Then \(T\) maps \(\cD_\rho\) bijectively onto
\(\mathcal Q_\rho\), and \(\mathcal Q_\rho\) is contained in
\begin{equation}\label{eq:transformed-domain-inclusion}
 \{g\in H^1(\R^d):A^{1/2}g\in L^2(\R^d)\}.
\end{equation}
For every \(f\in\cD_\rho\), with \(g=Tf=f\psi\),
\begin{equation}\label{eq:closed-transform}
 \cE_\rho(f,f)
 =\int_{\R^d}\bigl\{|\nabla g(x)|^2+
   (A(x)-\be d)g(x)^2\bigr\}\dd x.
\end{equation}
\end{lemma}

\begin{proof}
We prove the closed identity in four steps.  First we show that the poles can
be removed at zero weighted form cost.  Then we perform the ordinary
integration by parts away from the poles.  Next we identify the completed
transformed domain.  Finally we explain why no extra distributional
point-mass term is present in the closed realization.

So, first combine coincident poles, and write the resulting distinct poles as
\(p_1,\ldots,p_J\), with combined charges \(\bar\gamma_k>0\).  Fix
\(v\in C_c^\infty(\R^d)\).  Choose a radial function
\(\eta\in C^\infty(\R^d;[0,1])\) such that \(\eta=0\) on \(B(0,1)\),
\(\eta=1\) on \(B(0,2)^c\), and \(\|\nabla\eta\|_\infty\le C_\eta\).
Set
\[
 \eta_\varepsilon(x)=\eta(x/\varepsilon).
\]
Then \(\eta_\varepsilon=0\) on \(B(0,\varepsilon)\),
\(\eta_\varepsilon=1\) outside \(B(0,2\varepsilon)\),
\[
 \operatorname{supp}\nabla\eta_\varepsilon
 \subset\{x\in\R^d:\varepsilon<|x|<2\varepsilon\},
 \qquad
 |\nabla\eta_\varepsilon(x)|\le\frac{C_\eta}{\varepsilon}.
\]
For sufficiently small \(\varepsilon\), the balls about the distinct poles
are disjoint.  Since every term in \(U_\rho\) is nonnegative, on
\(B(p_k,2\varepsilon)\),
\begin{equation*}
 U_{\rho}(x)\ge \bar\gamma_k|x-p_k|^{2-d}.
\end{equation*}
Moreover, \(\ee^{-\be|x|^2}\le1\), and hence
\[
 \ee^{-H(x)}
 =\ee^{-\be|x|^2}\ee^{-U_\rho(x)}
 \le\exp[-\bar\gamma_k|x-p_k|^{2-d}]
 \qquad(x\in B(p_k,2\varepsilon)).
\]
Therefore, writing \(\omega_d=|B(0,1)|\),
\begin{align}
 \int_{\R^d} |v(x)|^2
  |\nabla\eta_\varepsilon(x-p_k)|^2\ee^{-H(x)}\dd x
 \le
 &\frac{C_\eta^2\|v\|_\infty^2}{\varepsilon^2}
 \int_{\{x\in\R^d:\varepsilon<|x-p_k|<2\varepsilon\}}
 \exp[-\bar\gamma_k|x-p_k|^{2-d}]\dd x
 \notag\\
 &\le
 \frac{C_\eta^2\|v\|_\infty^2}{\varepsilon^2}
 |B(0,2\varepsilon)|
 \exp[-2^{2-d}\bar\gamma_k\varepsilon^{2-d}]
 \notag\\
 &\le
 2^d\omega_d C_\eta^2\|v\|_\infty^2\varepsilon^{d-2}
 \exp[-2^{2-d}\bar\gamma_k\varepsilon^{2-d}]
 \longrightarrow0.
 \label{eq:pole-cutoff-cost}
\end{align}
Indeed, on the annulus \(|x-p_k|\le2\varepsilon\), and the exponent
\(2-d<0\), so
\[
 |x-p_k|^{2-d}\ge(2\varepsilon)^{2-d}
 =2^{2-d}\varepsilon^{2-d}.
\]
Also \(|B(0,2\varepsilon)|=\frac{2^d}{d}\sigma_{d-1}\varepsilon^d\).  Finally, with
\(t=\varepsilon^{-(d-2)}\), the last \(\varepsilon\)-dependent factor is
\[
 \varepsilon^{d-2}
 \ee^{-2^{2-d}\bar\gamma_k\varepsilon^{2-d}}
 =t^{-1}\ee^{-2^{2-d}\bar\gamma_k t}\longrightarrow0 \text{ as }\varepsilon \to 0.
\]

Let
\(\chi_\varepsilon(x)=\prod_{k=1}^J
\eta_\varepsilon(x-p_k)\) and
\(v_\varepsilon=v\chi_\varepsilon\).  Then
\(v_\varepsilon\in\mathcal C_\rho\), because it vanishes in a neighborhood
of every pole.  For every \(x\notin\operatorname{supp}\rho\), one has
\(\chi_\varepsilon(x)=1\) for all sufficiently small \(\varepsilon\), so
\(v_\varepsilon(x)\to v(x)\).  Since \(0\le\chi_\varepsilon\le1\),
\(|v_\varepsilon-v|\le|v|\), and dominated convergence gives
\[
 \|v_\varepsilon-v\|_{L^2(\R^d,\psi^2\dd x)}\longrightarrow0.
\]
Moreover,
\[
 \nabla(v_\varepsilon-v)
 =(\chi_\varepsilon-1)\nabla v+v\nabla\chi_\varepsilon.
\]
For the first term, dominated convergence gives
\[
 \int_{\R^d}|\chi_\varepsilon(x)-1|^2|\nabla v(x)|^2
 \psi(x)^2\dd x\longrightarrow0.
\]
For the second, the product rule gives
\[
 \nabla\chi_\varepsilon(x)
 =\sum_{k=1}^J\nabla\eta_\varepsilon(x-p_k)
   \prod_{\ell\ne k}\eta_\varepsilon(x-p_\ell).
\]
Because \(0\le\eta_\varepsilon\le1\), Cauchy--Schwarz implies $|\nabla\chi_\varepsilon(x)|^2
 \le J\sum_{k=1}^J
 |\nabla\eta_\varepsilon(x-p_k)|^2.$
Using \(\psi^2=\mathcal Z^{-1}\ee^{-H}\) and
\eqref{eq:pole-cutoff-cost}, we obtain
\begin{align*}
 \int_{\R^d}|v(x)|^2|\nabla\chi_\varepsilon(x)|^2
 \psi(x)^2\dd x
 &\le\frac{J}{\mathcal Z}\sum_{k=1}^J
 \int_{\R^d}|v(x)|^2|\nabla\eta_\varepsilon(x-p_k)|^2
 \ee^{-H(x)}\dd x\longrightarrow0.
\end{align*}
Finally, \(|a+b|^2\le2|a|^2+2|b|^2\) shows that
\[
 \cE_\rho(v_\varepsilon-v,v_\varepsilon-v)\longrightarrow0.
\]
Together with the \(L^2\)-convergence, this proves convergence in the form
norm for every compactly supported smooth \(v\).  Since the original closed
form was defined as the closure of \(C_c^\infty(\R^d)\), this proves that the
punctured space \(\mathcal C_\rho\) is a core.  Coincident poles cause no
additional case: they were combined at the start into the positive charges
\(\bar\gamma_k\).

Now, for \(f\in\mathcal C_\rho\), put \(g=f\psi\).  Its support has positive
distance from every pole, so ordinary integration by parts and
\(\Delta U_{\rho}=0\) on that support give
\begin{align*}
 \int_{\R^d}|\nabla f(x)|^2\psi(x)^2\dd x
 &=\int_{\R^d}\left|\nabla g(x)+\frac12g(x)\nabla H(x)\right|^2\dd x
 \\
 &=\int_{\R^d}\left(|\nabla g(x)|^2+
   \frac14|\nabla H(x)|^2g(x)^2\right)\dd x+\frac12\int_{\R^d}
   \nabla(g(x)^2)\cdot\nabla H(x)\dd x
 \\
 &=\int_{\R^d}|\nabla g(x)|^2\dd x
   +\int_{\R^d}\left(\frac14|\nabla H(x)|^2-
     \frac12\Delta H(x)\right)g(x)^2\dd x\\
 &=\int_{\R^d}\{|\nabla g(x)|^2+(A(x)-\be d)g(x)^2\}\dd x.
\end{align*}
In the last line we used
\(\Delta H=\Delta(\be|x|^2)+\Delta U_\rho=2\be d\) on the support of
\(g\).
\emph{Step 3: closing the identity.}
After adding \(\be d\|f\|_{L^2(\psi^2\dd x)}^2\), the two sides become the
same nonnegative norm,
\begin{equation}\label{eq:positive-transform-norm}
 \cE_\rho(f,f)+\be d\|f\|_2^2
 =\int_{\R^d}(|\nabla g(x)|^2+A(x)g(x)^2)\dd x.
\end{equation}
The map \(T\) is unitary because
\[
 \|Tf\|_{L^2(\R^d,\dd x)}^2
 =\int_{\R^d}|f(x)|^2\psi(x)^2\dd x.
\]
The left side of \eqref{eq:positive-transform-norm} is an equivalent norm
for the weighted form closure; explicitly, if
\(\|f\|_{\cD_\rho}^2=\cE_\rho(f,f)+\|f\|_2^2\), then
\[
 \min\{1,\be d\}\|f\|_{\cD_\rho}^2
 \le\cE_\rho(f,f)+\be d\|f\|_2^2
 \le\max\{1,\be d\}\|f\|_{\cD_\rho}^2.
\]
The right side of \eqref{eq:positive-transform-norm} is exactly
\eqref{eq:transformed-positive-norm}.  Since \(T\) is unitary on the
underlying \(L^2\) spaces, \eqref{eq:positive-transform-norm} also gives,
on \(T\mathcal C_\rho\),
\[
 \be d\|g\|_{L^2(\R^d)}^2
 \le\|g\|_{\mathcal Q_\rho}^2.
\]
Thus the abstract completion \(\mathcal Q_\rho\) is canonically realized as
a subspace of \(L^2(\R^d)\).  The map \(T\) sends the completion of
\(\mathcal C_\rho\) under the equivalent norm
\((\cE_\rho(f,f)+\be d\|f\|_2^2)^{1/2}\) isometrically onto
\(\mathcal Q_\rho\).  In particular, a form-Cauchy
sequence is carried to a Cauchy sequence in \(H^1\) and in
\(L^2(\R^d,A\dd x)\).  This proves the asserted onto identification,
\eqref{eq:transformed-domain-inclusion}, and \eqref{eq:closed-transform}.

Now, we finally clarify the distributional issue at the poles.  If
\(\sigma_{d-1}=|\mathbb S^{d-1}|\), then, in distributions,
\[
 -\Delta U_\rho=(d-2)\sigma_{d-1}\rho.
\]
Thus a purely formal use of the distributional Laplacian in the expansion
\(\frac14|\nabla H|^2-\frac12\Delta H\) might suggest an additional
quadratic form concentrated on \(\operatorname{supp}\rho\).  That is not how
the transformed form is defined.  The integration by parts is first carried
out for \(f\in\mathcal C_\rho\), whose support stays away from every pole, so
only the classical identity \(\Delta U_\rho=0\) is used.  The transformed
form is then the closure of this punctured-domain identity.  Estimate
\eqref{eq:pole-cutoff-cost} proves that removing the poles has zero weighted
form cost, so this closure is exactly the image of the original weighted
form and contains no separate quadratic-form contribution supported at the
poles.
\end{proof}

\begin{remark}[Langevin generator and the singular Witten Laplacian]
For a smooth potential \(H\), the nonnegative Langevin operator
\[
 L_H=-\Delta+\nabla H\cdot\nabla
\]
is symmetric in \(L^2(\ee^{-H}\dd x)\).  Multiplication by
\(\ee^{-H/2}\) formally gives the Witten--Schr\"odinger operator
\[
 \ee^{-H/2}L_H\ee^{H/2}
 =
 -\Delta+\frac14|\nabla H|^2-\frac12\Delta H.
\]
In the present setting,
\[
 H(x)=\be|x|^2+U_\rho(x),
 \qquad
 -\Delta U_\rho=(d-2)|\mathbb S^{d-1}|\rho
\]
in the sense of distributions.  Consequently, a purely formal insertion
of the distributional Laplacian would produce
\[
 \frac14|\nabla H|^2-\be d
 +\frac{d-2}{2}|\mathbb S^{d-1}|\rho.
\]
The last term is supported at the Coulomb poles.

Lemma~\ref{lem:singular-transform} does not simply discard this
distributional term.  Instead, it identifies the closed realization selected
by the original Langevin Dirichlet form.  The transform identity is first
proved on $C_c^\infty(\R^d\setminus\operatorname{supp}\rho)$,
where \(U_\rho\) is harmonic and hence
\(\Delta H=2\be d\) classically.  The pole-cutoff estimate then shows that
this punctured test-function class is a core for the weighted Langevin form.
Closing the identity therefore gives
\[
 \cE_\rho(f,f)
 =
 \int_{\R^d}
 \left\{
 |\nabla g(x)|^2+
 \left(\frac14|\nabla H(x)|^2-\be d\right)g(x)^2
 \right\}\dd x,
 \qquad g=f\psi,
\]
without an additional quadratic-form contribution supported on
\(\operatorname{supp}\rho\).

Thus the conclusion is not that the distributional identity for
\(\Delta U_\rho\) is false.  Rather, the closed Witten-Laplacian realization
unitarily equivalent to the Langevin form is obtained by closing the
punctured-domain expression, and in this realization the distributional
point masses do not define a separate term such as
\(\sum_k\gamma_k|g(y_k)|^2\).
\end{remark}

\begin{theorem}\label{thm:onesite}
For every \(d\ge3\) and \(\be>0\), there are constants
\begin{equation*}
 \eta_{d,\be}>0,
 \qquad
 K_{d,\be}<\infty
\end{equation*}
such that \(\Gamma(\rho)\le\eta_{d,\be}\) implies, for every
\(f\in\cD_{\rho}\),
\begin{equation*}
 \Var_{\nu_{\rho}}(f)
 \le K_{d,\be}\int_{\R^d}|\nabla f(x)|^2\dd\nu_{\rho}(x).
\end{equation*}
The constants are independent of the number, positions, and collisions of
the poles.
\end{theorem}

The proof is by contradiction.  This reduction is useful because it turns a
failure of uniform Poincar\'e constants into a sequence of normalized
near-ground states.  If the constants were not uniform as the total charge
tends to zero, one could find functions with unit variance, zero mean, and
vanishing Dirichlet energy for measures converging in total variation to the
Gaussian.  The ground-state transform converts these functions into almost
minimizers of a Schr\"odinger form.  The weak-type bound on the Coulomb field
then gives compactness and identifies the limiting form as the Gaussian
harmonic oscillator.  Its ground state is unique and not orthogonal to itself,
which contradicts the centering condition.

We prove Theorem~\ref{thm:onesite} after establishing four auxiliary
lemmas.

\begin{lemma}
\label{lem:normalized-counterexample}
Assume that for every \(\eta>0\) and every finite \(K\), there are a point
measure \(\rho\) with \(\Gamma(\rho)\le\eta\) and a function
\(h\in\cD_\rho\) such that
\[
 \Var_{\nu_\rho}(h)
 >K\int_{\R^d}|\nabla h(x)|^2\dd\nu_\rho(x).
\]
Then there are point measures \(\rho_n\) and functions
\(f_n\in\cD_{\rho_n}\) such that, with
\(\nu_n=\nu_{\rho_n}\) and \(\Gamma_n=\Gamma(\rho_n)\),
\[
 \Gamma_n\le\frac1n,
 \qquad
 \nu_n(f_n)=0,
 \qquad
 \nu_n(f_n^2)=1,
 \qquad
 E_n:=\int_{\R^d}|\nabla f_n(x)|^2\dd\nu_n(x)<\frac1n.
\]
\end{lemma}

\begin{proof}
Constants belong to every form domain \(\cD_\rho\).  This is needed because
we shall center the counterexample functions inside the closed form domain.
Indeed, let
\(\chi_R\in C_c^\infty(\R^d)\) satisfy \(0\le\chi_R\le1\),
\(\chi_R=1\) on \(B(0,R)\), \(\chi_R=0\) outside \(B(0,2R)\), and
\(|\nabla\chi_R|\le C/R\).  Since \(\nu_\rho\) is a probability,
\begin{align*}
 \int_{\R^d}|\chi_R(x)-1|^2\dd\nu_\rho(x)&\longrightarrow0,\\
 \int_{\R^d}|\nabla\chi_R(x)|^2\dd\nu_\rho(x)
 &\le\frac{C^2}{R^2}
 \nu_\rho\bigl(B(0,2R)\setminus B(0,R)\bigr)
 \le\frac{C^2}{R^2}\longrightarrow0.
\end{align*}
Thus \(\chi_R\to1\) in the form norm, so \(1\in\cD_\rho\) and
\(\cE_\rho(1,1)=0\).  Centering therefore preserves the form domain and
does not change the gradient.

Apply the assumption with \(\eta=1/n\) and \(K=n\).  This gives
\(\rho_n\) and \(h_n\in\cD_{\rho_n}\) such that
\[
 \Gamma_n\le\frac1n,
 \qquad
 0\le\int_{\R^d}|\nabla h_n(x)|^2\dd\nu_n(x)
 <\frac1n\Var_{\nu_n}(h_n).
\]
In particular, \(\sigma_n^2:=\Var_{\nu_n}(h_n)>0\).  Set
\[
 m_n:=\nu_n(h_n),
 \qquad
 f_n:=\frac{h_n-m_n}{\sigma_n}.
\]
Then
\[
 \nu_n(f_n)=\frac{\nu_n(h_n)-m_n}{\sigma_n}=0,
 \qquad
 \nu_n(f_n^2)=\frac{\Var_{\nu_n}(h_n)}{\sigma_n^2}=1,
\]
and \(\nabla f_n=\sigma_n^{-1}\nabla h_n\) almost everywhere.  Therefore
\[
 E_n
 =\frac1{\sigma_n^2}
 \int_{\R^d}|\nabla h_n(x)|^2\dd\nu_n(x)
 <\frac1n.
\]
\end{proof}

\begin{lemma}
\label{lem:transformed-compactness}
Let \((\rho_n)\) be finite positive point measures with
\(\Gamma_n:=\Gamma(\rho_n)\to0\).  Write
\[
 U_n=U_{\rho_n},\qquad
 H_n(x)=\be|x|^2+U_n(x),\qquad
 \mathcal Z_n=\int_{\R^d}\ee^{-H_n(x)}\dd x,
\]
and set
\[
 \psi_n(x)=\mathcal Z_n^{-1/2}\ee^{-H_n(x)/2},
 \qquad
 A_n(x)=\frac14|2\be x+\nabla U_n(x)|^2.
\]
For \(g\in\mathcal Q_{\rho_n}\), define
\begin{equation}\label{eq:ground-state-form}
 Q_n(g):=\int_{\R^d}\left\{|\nabla g(x)|^2+
 (A_n(x)-\be d)g(x)^2\right\}\dd x.
\end{equation}
Suppose that \(g_n\in\mathcal Q_{\rho_n}\), \(\|g_n\|_2=1\), and
\(Q_n(g_n)\to0\).  Then, after passage to a subsequence, there is
\(g\in H^1(\R^d)\) such that
\[
 g_n\rightharpoonup g\quad\text{in }H^1(\R^d),
 \qquad
 g_n\longrightarrow g\quad\text{in }L^2(\R^d),
 \qquad
 \|g\|_2=1.
\]
Moreover, if
\[
 \mathcal Z_0=\int_{\R^d}\ee^{-\be|x|^2}\dd x,
 \qquad
 \psi_0(x)=\mathcal Z_0^{-1/2}\ee^{-\be|x|^2/2},
\]
then \(\psi_n\to\psi_0\) in \(L^2(\R^d)\).
\end{lemma}

\begin{proof}
By \eqref{eq:ground-state-form},
\begin{equation}\label{eq:H1-A-bound}
 \int_{\R^d}|\nabla g_n(x)|^2\dd x+
 \int_{\R^d}A_n(x)g_n(x)^2\dd x
 =Q_n(g_n)+\be d\le\be d+o(1).
\end{equation}
Thus \((g_n)\) is bounded in \(H^1(\R^d)\), as $A_n \ge 0$.  If \(S_d\) denotes a constant
in the homogeneous Sobolev inequality, then
\[
 \|g_n\|_{2d/(d-2)}^2
 \le S_d\int_{\R^d}|\nabla g_n(x)|^2\dd x\le C_{d,\be}.
\]

For \(R>0\), set
\[
 B_n(R)=\{x\in\R^d:|x|>R,\ |\nabla U_n(x)|>\be|x|\}.
\]
Lemma~\ref{lem:weak-field}, applied with \(t=\be R\), and H\"older's
inequality yield
\begin{align*}
 |B_n(R)|
 &\le C_d\left(\frac{\Gamma_n}{\be R}\right)^{d/(d-1)},\\
 \int_{B_n(R)}g_n(x)^2\dd x
 &\le |B_n(R)|^{2/d}\|g_n\|_{2d/(d-2)}^2
 \le C_{d,\be}|B_n(R)|^{2/d}.
\end{align*}
On \(\{x\in\R^d:|x|>R\}\setminus B_n(R)\),
\[
 |2\be x+\nabla U_n(x)|
 \ge2\be|x|-|\nabla U_n(x)|
 \ge\be|x|\ge\be R,
\]
so \(A_n(x)\ge\be^2R^2/4\).  This is the only place where confinement is
used in the compactness argument: away from the exceptional set where the
Coulomb field is large, the quadratic trap forces mass back into a fixed
large ball.  Therefore
\[
 \int_{\{x\in\R^d:|x|>R\}}g_n(x)^2\dd x
 \le\int_{B_n(R)}g_n(x)^2\dd x
 +\frac4{\be^2R^2}\int_{\R^d}A_n(x)g_n(x)^2\dd x.
\]
For fixed \(R\), the first term tends to zero because \(\Gamma_n\to0\),
whereas \eqref{eq:H1-A-bound} bounds the second.  Hence
\begin{equation*}
 \limsup_{n\to\infty}
 \int_{\{x\in\R^d:|x|>R\}}g_n(x)^2\dd x
 \le\frac{C_{d,\be}}{R^2}.
\end{equation*}

Since \((g_n)\) is bounded in \(H^1(\R^d)\), reflexivity gives a
subsequence, not relabelled, and a function \(g\in H^1(\R^d)\) such that
\[
 g_n\rightharpoonup g
 \qquad\text{weakly in }H^1(\R^d).
\]
For every integer \(m\ge1\), the restrictions of \(g_n\) to \(B(0,m)\)
are bounded in \(H^1(B(0,m))\).  By the Rellich--Kondrachov theorem, a
subsequence converges strongly in \(L^2(B(0,m))\).  Applying this
successively for \(m=1,2,\ldots\) and taking a diagonal subsequence, we may
assume that
\[
 g_n\longrightarrow g
 \qquad\text{strongly in }L^2(B(0,R))
\]
for every \(R>0\).  The local strong limit agrees with the global weak
\(H^1\) limit because weak limits are unique.

For fixed \(R>0\), multiplication by
\(\mathbf1_{\{|x|>R\}}\) is a bounded operator on \(L^2(\R^d)\).
Consequently, weak lower semicontinuity of the \(L^2\)-norm gives
\[
 \int_{\{x\in\R^d:|x|>R\}}g(x)^2\dd x
 \le
 \liminf_{n\to\infty}
 \int_{\{x\in\R^d:|x|>R\}}g_n(x)^2\dd x
 \le \frac{C_{d,\be}}{R^2}.
\]
Thus both the sequence and its weak limit have uniformly small
\(L^2\)-tails.  Indeed,
\begin{align*}
 \int_{\R^d}|g_n(x)-g(x)|^2\dd x
 &=
 \int_{B(0,R)}|g_n(x)-g(x)|^2\dd x+
 \int_{\{x\in\R^d:|x|>R\}}|g_n(x)-g(x)|^2\dd x\\
 &\le
 \int_{B(0,R)}|g_n(x)-g(x)|^2\dd x+
 2\int_{\{x\in\R^d:|x|>R\}}(g_n(x)^2+g(x)^2)\dd x.
\end{align*}
Taking the upper limit as \(n\to\infty\), using the strong convergence on
\(B(0,R)\) and the two tail estimates, yields
\[
 \limsup_{n\to\infty}
 \int_{\R^d}|g_n(x)-g(x)|^2\dd x
 \le \frac{4C_{d,\be}}{R^2}.
\]
Letting \(R\to\infty\) proves that
\[
 g_n\longrightarrow g
 \qquad\text{strongly in }L^2(\R^d).
\]
Since \(\|g_n\|_2=1\) for every \(n\), strong \(L^2\) convergence finally
implies
\[
 \|g\|_2=\lim_{n\to\infty}\|g_n\|_2=1.
\]

It remains to prove convergence of the normalized ground states.  Since
\(0\le1-\ee^{-U_n}\le U_n\), Lemma~\ref{lem:gaussian-moment} gives
\begin{align*}
 0\le\mathcal Z_0-\mathcal Z_n
 &=\int_{\R^d}\ee^{-\be|x|^2}
   (1-\ee^{-U_n(x)})\dd x\le\int_{\R^d}\ee^{-\be|x|^2}U_n(x)\dd x
 \le C_{d,\be}\Gamma_n\longrightarrow0.
\end{align*}
For
\(\widetilde\psi_n=\ee^{-(\be|x|^2+U_n)/2}\) and
\(\widetilde\psi_0=\ee^{-\be|x|^2/2}\), the inequality
\((1-\ee^{-u/2})^2\le1-\ee^{-u}\) yields
\[
 \|\widetilde\psi_n-\widetilde\psi_0\|_2^2
 =\int_{\R^d}\ee^{-\be|x|^2}
   (1-\ee^{-U_n(x)/2})^2\dd x
 \le\mathcal Z_0-\mathcal Z_n\longrightarrow0.
\]
Because \(\mathcal Z_n\to\mathcal Z_0>0\), normalization gives
\(\psi_n\to\psi_0\) in \(L^2(\R^d)\).
\end{proof}

\begin{lemma}
\label{lem:oscillator-liminf}
Let \((\rho_n)\) be finite positive point measures satisfying
$ \Gamma_n:=\Gamma(\rho_n)\to 0.$
and set
\[
 U_n:=U_{\rho_n},
 \qquad
 A_n(x):=\frac14|2\be x+\nabla U_n(x)|^2.
\]
For \(h\in\mathcal Q_{\rho_n}\) as was defined in Lemma \ref{lem:singular-transform}, define
\[
 Q_n(h):=
 \int_{\R^d}
 \left\{
 |\nabla h(x)|^2+
 (A_n(x)-\be d)h(x)^2
 \right\}\dd x.
\]
Suppose that \(g_n\in\mathcal Q_{\rho_n}\), that
\[
 \|g_n\|_2=1,
 \qquad
 \sup_n Q_n(g_n)<\infty,
\]
and that \((g_n)\) is bounded in \(H^1(\R^d)\).  Assume moreover that
\[
 g_n\rightharpoonup g
 \quad\text{weakly in }H^1(\R^d),
 \qquad
 g_n\longrightarrow g
 \quad\text{strongly in }L^2(\R^d).
\]
Then
\[
 \liminf_{n\to\infty}
 \int_{\R^d}A_n(x)g_n(x)^2\dd x
 \ge
 \int_{\R^d}\be^2|x|^2g(x)^2\dd x.
\]
In particular, \(|x|g\in L^2(\R^d)\).  If
\[
 Q_0(g):=
 \int_{\R^d}
 \left\{
 |\nabla g(x)|^2+
 (\be^2|x|^2-\be d)g(x)^2
 \right\}\dd x,
\]
then
\[
 Q_0(g)\le\liminf_{n\to\infty}Q_n(g_n).
\]
\end{lemma}

\begin{proof}
Fix \(R,\delta>0\). By definition,
\[
 F_n(\delta)
 =\{x\in\R^d:|\nabla U_{\rho_n}(x)|>\delta\}.
\]
Applying Lemma~\ref{lem:weak-field} to the point measure \(\rho_n\) with
threshold \(t=\delta\) gives
\[
 |F_n(\delta)|
 \le
 C_d\left(\frac{\Gamma(\rho_n)}{\delta}\right)^{d/(d-1)}
 =
 C_d\left(\frac{\Gamma_n}{\delta}\right)^{d/(d-1)}.
\]
Since \(\delta>0\) is fixed and \(\Gamma_n\to0\), the right-hand side
converges to zero.
The Sobolev inequality and the uniform \(H^1\) bound imply
\begin{equation}\label{eq:large-field-mass}
 \int_{F_n(\delta)}g_n(x)^2\dd x
 \le |F_n(\delta)|^{2/d}\|g_n\|_{2d/(d-2)}^2
 \longrightarrow0.
\end{equation}
On \(B(0,R)\setminus F_n(\delta)\),
\[
 A_n(x)=\be^2|x|^2+\be x\cdot\nabla U_n(x)
 +\frac14|\nabla U_n(x)|^2
 \ge\be^2|x|^2-\be R\delta.
\]
Since \(A_n\ge0\),
\begin{align*}
 \int_{\R^d}A_n(x)g_n(x)^2\dd x
 &\ge\int_{B(0,R)\setminus F_n(\delta)}
   \be^2|x|^2g_n(x)^2\dd x-\be R\delta
   \int_{B(0,R)\setminus F_n(\delta)}g_n(x)^2\dd x\\
 &\ge\int_{B(0,R)\setminus F_n(\delta)}
   \be^2|x|^2g_n(x)^2\dd x-\be R\delta.
\end{align*}
Strong local \(L^2\) convergence and \eqref{eq:large-field-mass} show that
the remaining integral converges to
\(\int_{B(0,R)}\be^2|x|^2g(x)^2\dd x\).  Thus
\[
 \liminf_{n\to\infty}\int_{\R^d}A_n(x)g_n(x)^2\dd x
 \ge\int_{B(0,R)}\be^2|x|^2g(x)^2\dd x-\be R\delta.
\]
Letting first \(\delta\downarrow0\) and then \(R\uparrow\infty\) proves
the potential liminf.  It also shows that \(|x|g\in L^2(\R^d)\).  Weak
lower semicontinuity of the Dirichlet integral, strong \(L^2\) convergence,
and \eqref{eq:ground-state-form} now give the asserted inequality for
\(Q_0\).
\end{proof}

\begin{lemma}\label{lem:oscillator-kernel}
Let \(Q_0\) be the quadratic form defined in
Lemma~\ref{lem:oscillator-liminf}.  For every \(g\) in its form domain,
\[
 Q_0(g)=\int_{\R^d}|\nabla g(x)+\be xg(x)|^2\dd x\ge0.
\]
Equality holds if and only if \(g\) is a scalar multiple of \(\psi_0\).
\end{lemma}

\begin{proof}
For \(g\in C_c^\infty(\R^d)\), expansion and integration by parts give
\begin{align*}
 \int_{\R^d}|\nabla g(x)+\be xg(x)|^2\dd x
 &=\int_{\R^d}\bigl(|\nabla g(x)|^2+\be^2|x|^2g(x)^2\bigr)\dd x+\be\int_{\R^d}x\cdot\nabla(g(x)^2)\dd x\\
 &=\int_{\R^d}\bigl(|\nabla g(x)|^2+
   (\be^2|x|^2-\be d)g(x)^2\bigr)\dd x,
\end{align*}
because
\(\int_{\R^d}x\cdot\nabla(g(x)^2)\dd x
=-d\int_{\R^d}g(x)^2\dd x\).  Closure extends the identity to the
oscillator form domain, cf. Lemma~\ref{lem:singular-transform}.  If equality holds, then
\(\nabla g+\be xg=0\) distributionally.  Hence
\(\nabla(\ee^{\be|x|^2/2}g)=0\), so \(g\) is a scalar multiple of
\(\ee^{-\be|x|^2/2}\), equivalently of \(\psi_0\).  The converse follows
by direct substitution.
\end{proof}

An application of the four lemmas above yields the result. Indeed: 

\begin{proof}[Proof of Theorem~\ref{thm:onesite}]
Suppose otherwise.  The negation of the theorem is precisely the hypothesis
of Lemma~\ref{lem:normalized-counterexample}.  Choose the resulting sequences
\((\rho_n,f_n)\).  Then \(\Gamma_n\to0\), \(E_n\to0\),
\(\nu_n(f_n)=0\), and \(\nu_n(f_n^2)=1\).

With the notation of Lemma~\ref{lem:transformed-compactness}, set
\(g_n=f_n\psi_n\).  Then
\[
 \|g_n\|_2^2
 =\int_{\R^d}f_n(x)^2\psi_n(x)^2\dd x
 =\nu_n(f_n^2)=1
\text{ and }
 \langle g_n,\psi_n\rangle
 =\int_{\R^d}f_n(x)\psi_n(x)^2\dd x
 =\nu_n(f_n)=0.
\]
Lemma~\ref{lem:singular-transform} and \eqref{eq:closed-transform} give
\[
 Q_n(g_n)=\cE_{\rho_n}(f_n,f_n)=E_n\longrightarrow0.
\]
Lemma~\ref{lem:transformed-compactness} therefore yields, after extraction,
\[
 g_n\rightharpoonup g\ \text{in }H^1(\R^d),\qquad
 g_n\to g\ \text{in }L^2(\R^d),\qquad
 \|g\|_2=1,
\]
and \(\psi_n\to\psi_0\) in \(L^2(\R^d)\).  Hence
\[
 |\langle g,\psi_0\rangle|
 \le\|g-g_n\|_2\|\psi_0\|_2
 +\|g_n\|_2\|\psi_0-\psi_n\|_2\longrightarrow0,
\]
so \(g\perp\psi_0\).

Lemma~\ref{lem:oscillator-liminf} gives
\[
 Q_0(g)\le\liminf_{n\to\infty}Q_n(g_n)=0.
\]
Lemma~\ref{lem:oscillator-kernel} gives \(Q_0(g)\ge0\), with equality only
for scalar multiples of \(\psi_0\).  Thus \(g\in\operatorname{span}\{\psi_0\}\).
Since \(g\perp\psi_0\), this forces \(g=0\), contradicting \(\|g\|_2=1\).
\end{proof}

\begin{remark}[Role of Coulomb harmonicity]\label{rem:coulomb-harmonicity}
The identity \(\Delta U=0\) off the poles is the Coulomb-specific input in
\eqref{eq:closed-transform}.  No geometric regularity or separation of the
poles enters the argument.
\end{remark}

\section{Moving one pole and Dobrushin tensorization}

We turn to the dependence of a one-site law on a single added or moved pole.
Only the normalization and moment bounds from Proposition
\ref{prop:basic-onesite} enter the estimate.

\begin{proposition}
\label{prop:moving-center}
Let \(\rho\) be a finite positive point measure, and let \(a\ge0\) be the
charge of an additional Coulomb pole.  Assume that for $M_{d,\beta}$ the constant in Lemma \ref{lem:gaussian-moment}
\[
 \Gamma(\rho)+a\le\frac{1}{2M_{d,\be}}.
\]
For \(y\in\R^d\), define
\[
 \nu^y:=\nu_{\rho+a\delta_y}, \qquad\text{i.e. }
 \nu^y(dx)
 =
 \frac{\ee^{-a|x-y|^{2-d}}}{
   \int_{\R^d}\ee^{-a|u-y|^{2-d}}\dd\nu_\rho(u)}
 \nu_\rho(dx).
\]
Then, for every \(y\in\R^d\),
\begin{equation}\label{eq:add-remove-TV}
 \|\nu^y-\nu_\rho\|_{\TV}
 \le4M_{d,\be}a.
\end{equation}
Moreover, for every \(y,y'\in\R^d\),
\begin{equation}\label{eq:move-TV}
 \|\nu^y-\nu^{y'}\|_{\TV}
 \le8M_{d,\be}a.
\end{equation}
\end{proposition}

\begin{proof}
We first prove the add-remove estimate \eqref{eq:add-remove-TV} for a fixed
pole location \(y\).  The proof uses a bounded truncation of the singular
kernel, differentiates along the interpolation that turns on the pole, and
then removes the truncation by dominated convergence.

Fix \(y\in\R^d\).  For \(R>0\), define the bounded truncation
\[
 \phi_R(x):=\min\{|x-y|^{2-d},R\},
 \qquad x\in\R^d,
\]
where we set \(\phi_R(y)=R\).  For \(0\le t\le1\), set
\[
 Z_{t,R}
 :=
 \int_{\R^d}\ee^{-ta\phi_R(u)}\dd\nu_\rho(u) \text{ and }
 \nu_{t,R}(d x)
 :=
 Z_{t,R}^{-1}\ee^{-ta\phi_R(x)}\nu_\rho(d x).
\]
Since
\[
 \ee^{-taR}\le \ee^{-ta\phi_R(x)}\le1,
\]
one has \(0<Z_{t,R}\le1\), so \(\nu_{t,R}\) is a well-defined
probability measure.  Notice also that
\[
 \nu_{0,R}=\nu_\rho.
\]

We first obtain a bound on the \(\phi_R\)-moment of \(\nu_{t,R}\).
Because \(0\le\phi_R\le R\), differentiation under the integral is
justified by dominated convergence.  Thus
\begin{align*}
 \frac{\dd}{\dd t}Z_{t,R}
 &=-a\int_{\R^d}\phi_R(u)\ee^{-ta\phi_R(u)}
   \dd\nu_\rho(u)=-aZ_{t,R}\int_{\R^d}\phi_R(u)\dd\nu_{t,R}(u)
  =-aZ_{t,R}\nu_{t,R}(\phi_R).
\end{align*}

Define
\[
 N_{t,R}
 :=
 \int_{\R^d}\phi_R(u)\ee^{-ta\phi_R(u)}
 \dd\nu_\rho(u).
\]
Then $\nu_{t,R}(\phi_R)=\frac{N_{t,R}}{Z_{t,R}}.$ Again using the boundedness of \(\phi_R\), we may differentiate \(N_{t,R}\):
\begin{align*}
 \frac{\dd}{\dd t}N_{t,R}
 &=
 -a\int_{\R^d}\phi_R(u)^2\ee^{-ta\phi_R(u)}
 \dd\nu_\rho(u)=
 -aZ_{t,R}\nu_{t,R}(\phi_R^2).
\end{align*}
The quotient rule therefore gives
\begin{align*}
 \frac{\dd}{\dd t}\nu_{t,R}(\phi_R)
 &=
 \frac{N_{t,R}'Z_{t,R}-N_{t,R}Z_{t,R}'}
      {Z_{t,R}^2}=
 -a\nu_{t,R}(\phi_R^2)
 +a\nu_{t,R}(\phi_R)^2=
 -a\Var_{\nu_{t,R}}(\phi_R)
 \le0.
\end{align*}
Hence the map \(t\mapsto\nu_{t,R}(\phi_R)\) is nonincreasing.  Since
\(\nu_{0,R}=\nu_\rho\), it follows that
\[
 \nu_{t,R}(\phi_R)
 \le\nu_{0,R}(\phi_R)
 =\nu_\rho(\phi_R).
\]
Moreover,
\[
 0\le\phi_R(x)\le |x-y|^{2-d}
 \qquad\text{for }x\ne y.
\]
Since \(\nu_\rho\) is absolutely continuous with respect to Lebesgue
measure, the point \(y\) has zero \(\nu_\rho\)-measure.  Proposition
\ref{prop:basic-onesite} therefore gives
\begin{align*}
 \nu_{t,R}(\phi_R)
 &\le
 \int_{\R^d}\phi_R(x)\dd\nu_\rho(x)\le
 \int_{\R^d}|x-y|^{2-d}\dd\nu_\rho(x)\le 2M_{d,\be}.
\end{align*}
This estimate is uniform in \(t\in[0,1]\), \(R>0\), and \(y\in\R^d\). To estimate the change along the interpolation
\(t\mapsto\nu_{t,R}\), let \(h:\R^d\to\R\) be bounded and measurable.
Write
\[
 H_{t,R}
 :=
 \int_{\R^d}h(u)\ee^{-ta\phi_R(u)}\dd\nu_\rho(u),
\]
so that $\nu_{t,R}(h)=\frac{H_{t,R}}{Z_{t,R}}.$
Since \(h\) and \(\phi_R\) are bounded,
\begin{align*}
 \frac{\dd}{\dd t}H_{t,R}
 &=
 -a\int_{\R^d}h(u)\phi_R(u)
 \ee^{-ta\phi_R(u)}\dd\nu_\rho(u)=
 -aZ_{t,R}\nu_{t,R}(h\phi_R).
\end{align*}
Using the quotient rule and the preceding formula for \(Z_{t,R}'\), we
obtain
\begin{align*}
 \frac{\dd}{\dd t}\nu_{t,R}(h)
 &=
 \frac{H_{t,R}'Z_{t,R}-H_{t,R}Z_{t,R}'}
      {Z_{t,R}^2}=
 -a\nu_{t,R}(h\phi_R)
 +a\nu_{t,R}(h)\nu_{t,R}(\phi_R)=
 -a\operatorname{Cov}_{\nu_{t,R}}(h,\phi_R).
\end{align*}

Assume now that \(\|h\|_\infty\le1\).  Since \(\nu_{t,R}\) is a
probability measure $|\nu_{t,R}(h)|\le1$
and consequently
\[
 |h(x)-\nu_{t,R}(h)|\le2
 \qquad\text{for every }x\in\R^d.
\]
Because \(\phi_R\ge0\),
\begin{align*}
 \left|
 \operatorname{Cov}_{\nu_{t,R}}(h,\phi_R)
 \right|
 &=
 \left|
 \int_{\R^d}
 \bigl(h(x)-\nu_{t,R}(h)\bigr)\phi_R(x)
 \dd\nu_{t,R}(x)
 \right|\le
 \int_{\R^d}
 |h(x)-\nu_{t,R}(h)|\phi_R(x)
 \dd\nu_{t,R}(x)\\
 &\le
 2\int_{\R^d}\phi_R(x)\dd\nu_{t,R}(x)=
 2\nu_{t,R}(\phi_R)\le4M_{d,\be}.
\end{align*}
It follows that
\[
 \left|
 \frac{\dd}{\dd t}\nu_{t,R}(h)
 \right|
 \le4M_{d,\be}a.
\]
Integrating this estimate over \(t\in[0,1]\), and using
\(\nu_{0,R}=\nu_\rho\), gives
\begin{align*}
 |\nu_{1,R}(h)-\nu_\rho(h)|
 &=
 \left|
 \int_0^1
 \frac{\dd}{\dd t}\nu_{t,R}(h)\dd t
 \right|\le
 \int_0^1
 \left|
 \frac{\dd}{\dd t}\nu_{t,R}(h)
 \right|\dd t \le4M_{d,\be}a.
\end{align*}
Taking the supremum over all measurable \(h\) with
\(\|h\|_\infty\le1\) yields $\|\nu_{1,R}-\nu_\rho\|_{\TV}
 \le4M_{d,\be}a.$

It remains to remove the truncation.  Set $\phi(x):=|x-y|^{2-d}$ for $x\ne y.$
For every \(x\ne y\),
\[
 \phi_R(x)\uparrow\phi(x)
 \qquad\text{and hence}\qquad
 \ee^{-a\phi_R(x)}\longrightarrow\ee^{-a\phi(x)}
 \quad\text{as }R\to\infty.
\]
Moreover, $0\le\ee^{-a\phi_R(x)}\le1.$
Since \(\nu_\rho(\{y\})=0\), dominated convergence gives
\[
 Z_{1,R}
 =
 \int_{\R^d}\ee^{-a\phi_R(x)}\dd\nu_\rho(x)
 \longrightarrow
 Z_1
 :=
 \int_{\R^d}\ee^{-a|x-y|^{2-d}}\dd\nu_\rho(x).
\]
The limiting normalizing constant is strictly positive because the
integrand is strictly positive on \(\R^d\setminus\{y\}\), a set of full
\(\nu_\rho\)-measure.  Thus \(Z_1>0\).

The densities of \(\nu_{1,R}\) and \(\nu^y\) with respect to \(\nu_\rho\)
are
\[
 p_R(x):=Z_{1,R}^{-1}\ee^{-a\phi_R(x)},
 \qquad
 p(x):=Z_1^{-1}\ee^{-a|x-y|^{2-d}},
\]
respectively.  We have
\begin{align*}
 \int_{\R^d}|p_R(x)-p(x)|\dd\nu_\rho(x)
 &\le
 \left|Z_{1,R}^{-1}-Z_1^{-1}\right|
 \int_{\R^d}\ee^{-a\phi_R(x)}\dd\nu_\rho(x)+
 Z_1^{-1}
 \int_{\R^d}
 \left|
 \ee^{-a\phi_R(x)}
 -\ee^{-a|x-y|^{2-d}}
 \right|
 \dd\nu_\rho(x)\\
 &=
 Z_{1,R}\left|Z_{1,R}^{-1}-Z_1^{-1}\right|+Z_1^{-1}
 \int_{\R^d}
 \left|
 \ee^{-a\phi_R(x)}
 -\ee^{-a|x-y|^{2-d}}
 \right|
 \dd\nu_\rho(x).
\end{align*}
The first term converges to zero because \(Z_{1,R}\to Z_1>0\), and the
second converges to zero by dominated convergence.  Hence
\[
 \int_{\R^d}|p_R(x)-p(x)|\dd\nu_\rho(x)\longrightarrow0.
\]
Since both measures are absolutely continuous with respect to \(\nu_\rho\),
with densities \(p_R\) and \(p\), respectively, the definition of the full
total-variation norm gives
\begin{align*}
 \|\nu_{1,R}-\nu^y\|_{\TV}
 &=
 \sup_{\|h\|_\infty\le1}
 \left|
 \int_{\R^d}h(x)\bigl(p_R(x)-p(x)\bigr)\dd\nu_\rho(x)
 \right|\le
 \int_{\R^d}|p_R(x)-p(x)|\dd\nu_\rho(x).
\end{align*}
The \(L^1(\nu_\rho)\)-convergence of the densities established above now
implies
\[
 \|\nu_{1,R}-\nu^y\|_{\TV}\longrightarrow0.
\]
Therefore,  $\|\nu^y-\nu_\rho\|_{\TV}
 \le
 \|\nu^y-\nu_{1,R}\|_{\TV}
 +\|\nu_{1,R}-\nu_\rho\|_{\TV}.$
Letting \(R\to\infty\) gives $\|\nu^y-\nu_\rho\|_{\TV}
 \le4M_{d,\be}a,$
which proves \eqref{eq:add-remove-TV}.

Finally, applying \eqref{eq:add-remove-TV} separately to the poles at
\(y\) and \(y'\), and using the triangle inequality, gives
\begin{align*}
 \|\nu^y-\nu^{y'}\|_{\TV}
 &\le
 \|\nu^y-\nu_\rho\|_{\TV}
 +\|\nu_\rho-\nu^{y'}\|_{\TV}\le
 4M_{d,\be}a+4M_{d,\be}a=8M_{d,\be}a.
\end{align*}
This proves \eqref{eq:move-TV}.
\end{proof}
\begin{definition}
\label{def:Dobrushinmatrix}
Let \(T\) be a finite set and let \(\mu\) be a probability measure on
\((\R^d)^T\).

For \(i\in T\), write $x_{-i}:=(x_k)_{k\in T\setminus\{i\}}.$
For \(u\in\R^d\), let \(x^{i\to u}\) denote the configuration obtained
from \(x\) by replacing \(x_i\) with \(u\).

We choose, for every \(i\in T\), a version $\mu_i(\dd u\mid x_{-i})$
of the conditional distribution of \(x_i\) given \(x_{-i}\).  Thus, for
every bounded measurable \(G:(\R^d)^T\to\R\),
\[
 \int_{(\R^d)^T}G(x)\dd\mu(x)
 =
 \int_{(\R^d)^T}\int_{\R^d}
 G(x^{i\to u})\dd\mu_i(u\mid x_{-i})\dd\mu(x).
\]

For a measurable \(F:(\R^d)^T\to\R\), define its conditional variance at
site \(i\), with \(x_{-i}\) fixed, by
\begin{align*}
 \Var_i(F)(x_{-i})
 &:=
 \int_{\R^d}F(x^{i\to u})^2
 \dd\mu_i(u\mid x_{-i})-
 \left(
 \int_{\R^d}F(x^{i\to u})
 \dd\mu_i(u\mid x_{-i})
 \right)^2,
\end{align*}
whenever the two integrals are finite.

Equip \(\R^d\) with the trivial metric $d_0(u,v):=\mathbf1_{\{u\ne v\}}.$
For probability measures \(\sigma\) and \(\tau\) on \(\R^d\), define the $1$-Wasserstein distance
\[
 W_{1,d_0}(\sigma,\tau)
 :=
 \inf_{\pi\in\Pi(\sigma,\tau)}
 \int_{\R^d\times\R^d}d_0(u,v)\dd\pi(u,v),
\]
where \(\Pi(\sigma,\tau)\) denotes the set of couplings of \(\sigma\) and
\(\tau\).

For \(i\ne j\), define 
\begin{equation}\label{eq:wu-matrix}
 c_{ij}
 :=
 \sup
 \left\{
 W_{1,d_0}\bigl(
 \mu_i(\cdot\mid x_{-i}),
 \mu_i(\cdot\mid x'_{-i})
 \bigr):
 \begin{array}{l}
 x_{-i},x'_{-i}\in(\R^d)^{T\setminus\{i\}},\\
 x_k=x'_k\text{ for every }k\in T\setminus\{i,j\}
 \end{array}
 \right\},
\end{equation}
and set \(c_{ii}:=0\).  Thus \(c_{ij}\) is the largest change in the
conditional law at site \(i\) that can be caused by changing only the
boundary coordinate \(x_j\).

The matrix $C:=(c_{ij})_{i,j\in T}$
is called the \emph{Dobrushin interdependence matrix}.  Its spectral radius is
\[\rho(C) :=
 \max\{|\lambda|:\lambda\text{ is an eigenvalue of }C\}.
\]
\end{definition}

\underline{Relation with Wu’s interdependence coefficient in \cite{Wu}.} 
The coefficient $c_{ij}$ defined above in \eqref{eq:wu-matrix}, is precisely the interdependence coefficient defined in \cite[(2.3)]{Wu}  when the one-site space is equipped with the discrete metric $d_0(u,v):=\mathbf1_{\{u\ne v\}}.$ 
%Indeed, As usual for a Lipschitz quotient, Wu takes the supremum over boundary configurations differing only at site $j$ with \(x_j\ne x'_j\), and divides the Wasserstein distance by \(d_0(x_j, x_j')\). 
%The coefficient in \cite[(2.3)]{Wu} is exactly \eqref{eq:wu-matrix} for the metric \(d_0\).  
 Indeed, as usual for a Lipschitz quotient, Wu's supremum is taken
only over pairs
satisfying
\(
 d(x_j,x'_j)>0.\)
For the trivial metric \(d_0\), this condition is equivalent to
\(x_j\ne x'_j\).  Thus Wu's coefficient is
\[
 c_{ij}^{\mathrm{Wu}}
=\sup_{\substack{x,x'\in(\R^d)^T\\
 x_k=x'_k\text{ for every }k\ne j\\
 x_j\ne x'_j}}
 \frac{
 W_{1,d_0}\bigl(\mu_i(\cdot\mid x),
                 \mu_i(\cdot\mid x')\bigr)
 }{
 d_0(x_j,x'_j)
 }.
\]
Since \(d_0(x_j,x'_j)=1\) whenever \(x_j\ne x'_j\), this becomes
\[
 c_{ij}^{\mathrm{Wu}}
 =
 \sup_{\substack{x,x'\in(\R^d)^T\\
 x_k=x'_k\text{ for every }k\ne j\\
 x_j\ne x'_j}}
 W_{1,d_0}\bigl(\mu_i(\cdot\mid x),
                 \mu_i(\cdot\mid x')\bigr).
\]
If \(x_j=x'_j\), then \(x=x'\), so the numerator also vanishes.  Such
pairs are excluded from the Lipschitz quotient and do not affect the
supremum.

Moreover, for the trivial metric $d_0$, one has, by duality, see for instance \cite[below (2.1)]{Wu}
\begin{align} \label{eq: Wasserstein_TV}
 W_{1,d_0}(\sigma,\tau)
 =
 \sup_{A\in\mathcal B(\R^d)}|\sigma(A)-\tau(A)|
 =
 \frac12\|\sigma-\tau\|_{\TV}.
\end{align}
The proof of Proposition~\ref{prop:wu-criterion} recalls this identity.

We use the following form of Wu's tensorization result
\cite[Theorem 2.1]{Wu}.

\begin{proposition}[Dobrushin--Wu conditional tensorization]
\label{prop:wu-criterion}
Let \(T\) be finite, let \(\mu\) be a probability measure on
\((\R^d)^T\), and let \(C\) be its Dobrushin interdependence matrix as
defined in Def. \ref{def:Dobrushinmatrix}.

Assume that $\rho(C)<1.$
Suppose, moreover, that there is a constant \(K<\infty\), independent of
the site and of the boundary configuration, such that
\begin{equation}\label{eq:wu-conditional-PI}
 \Var_{\mu_i(\cdot\mid x_{-i})}(h)
 \le
 K\int_{\R^d}|\nabla h(u)|^2
 \dd\mu_i(u\mid x_{-i})
\end{equation}
for every \(i\in T\), every
\(x_{-i}\in(\R^d)^{T\setminus\{i\}}\), and every
\(h\in C_c^\infty(\R^d)\).

Then
\begin{equation}\label{eq:wu-gradient-PI}
 \Var_\mu(F)
 \le
 \frac{K}{1-\rho(C)}
 \sum_{i\in T}
 \int_{(\R^d)^T}|\nabla_iF(x)|^2\dd\mu(x)
\end{equation}
for every \(F\in C_c^\infty((\R^d)^T)\).  Here \(\nabla_iF\) denotes the
Euclidean gradient with respect to \(x_i\).

The inequality extends by form closure to the domain of the closed global
Dirichlet form.
\end{proposition}

\begin{proof}
For the trivial metric, every \(\pi\in\Pi(\sigma,\tau)\) has transportation
cost
\[
 \int_{\R^d\times\R^d}d_0(u,v)\dd\pi(u,v)=\pi\{u\ne v\}.
\]
The maximal-coupling formula therefore gives
\[
 W_{1,d_0}(\sigma,\tau)
 =\inf_\pi\pi\{u\ne v\}
 =\sup_A|\sigma(A)-\tau(A)|
 =\tfrac12\|\sigma-\tau\|_{\TV}.
\]
We aim to apply \cite[Theorem 2.1]{Wu} which requires a second moment assumption. This is automatic because \(d_0\le1\).  All
objects in this paragraph are purely measurable: the regular conditional laws are kernels on the Borel space \(\R^d\), and the coefficient above is
defined through couplings, so no differentiability or continuity of the
kernels is required.

We recall the definition of  the conditional variance of \(F\) at site \(i\), with
the boundary \(x_{-i}\) fixed, by
\begin{align*}
 \Var_i(F)(x_{-i})
 &:={}
 \int_{\R^d}F(x^{i\to u})^2\dd\mu_i(u\mid x_{-i})-
 \left(\int_{\R^d}F(x^{i\to u})\dd\mu_i(u\mid x_{-i})\right)^2.
\end{align*}
 Therefore \cite[Theorem 2.1, equation (2.4)]{Wu} gives
\begin{equation*}
 (1-\rho(C))\Var_\mu(F)
 \le\int_{(\R^d)^T}\sum_{i\in T}
 \Var_i(F)(x_{-i})\dd\mu(x).
\end{equation*}
Applying \eqref{eq:wu-conditional-PI} to the function
\(u\mapsto F(x^{i\to u})\), with \(x_{-i}\) fixed, gives
\begin{align*}
 \int_{(\R^d)^T}\sum_{i\in T}\Var_i(F)(x_{-i})\dd\mu(x)
 &\le K\int_{(\R^d)^T}\sum_{i\in T}\int_{\R^d}
   |\nabla_iF(x^{i\to u})|^2
   \dd\mu_i(u\mid x_{-i})\dd\mu(x)\\
 &=K\sum_{i\in T}\int_{(\R^d)^T}|\nabla_iF(x)|^2\dd\mu(x),
\end{align*}
where the last equality is the defining property of conditional
expectation.  Substitution in the preceding inequality proves
\eqref{eq:wu-gradient-PI}. 

This is the same two-step passage used for
Theorem 2.2 of \cite{Wu}; we record it separately because the Dobrushin
matrix here uses the trivial metric whereas the conditional form uses the
Euclidean gradient.  Theorem 2.1 explicitly allows a measurable one-site
space equipped with the trivial metric, so no topological separability of
\((\R^d,d_0)\) is required.
\end{proof}

\section{Many-particle proof}

Since our Poincar\'e bound will only be effective for $N$ large enough, we first prove that each fixed finite particle system has a positive, although not necessarily uniform or quantitative, Poincar\'e gap.  

Thus in the following lemma we treat the finitely many $N$-case by a fixed-particle gap observation.
\begin{lemma}
\label{lem:fixed-N-gap}
Fix \(N\ge2\), \(a\ge0\), and \(d\ge3\).  The probability measure
\begin{equation*}
 \mu_{N,a}(\dd z)=\mathcal Z_{N,a}^{-1}
 \exp\left[-\be\sum_{i=1}^N|z_i|^2
 -a\sum_{i<j}|z_i-z_j|^{2-d}\right]\dd z
\end{equation*}
where \(\mathcal Z_{N,a}\) is the normalizing constant, has a positive
Poincar\'e gap.
\end{lemma}

\begin{proof}
The case \(a=0\) is Gaussian, so assume \(a>0\).  We prove the positive gap
in four steps: first the collision strata have zero weighted form cost, then
we use the ground-state transform, then we prove compactness of the transformed
resolvent by showing that its potential is proper, and finally we use
connectedness to make the zero eigenvalue simple.

\emph{Step 1: zero form cost of collisions.} This is the many-particle analogue of the pole removal argument in
Lemma \ref{lem:singular-transform}.
There, cut-offs around fixed Coulomb
poles were shown to have vanishing weighted form cost. Here the singular
sets are the collision diagonals $\{z\in(\R^d)^N:z_i=z_j\}$. 
Put
\begin{equation*}
 \Omega_N=\{z:z_i\ne z_j\text{ for }i\ne j\},
 \qquad
 H(z)=\be\sum_i|z_i|^2+aW(z),
 \qquad
 W(z)=\sum_{i<j}|z_i-z_j|^{2-d}.
\end{equation*}
Fix a pair \(i<j\), put \(r=z_i-z_j\) and
\(c=(z_i+z_j)/2\), and use the same cutoff \(\eta_\varepsilon(r)\) as in
Lemma~\ref{lem:singular-transform}.  Since
\[
 |z_i|^2+|z_j|^2=2|c|^2+\tfrac12|r|^2
 \quad\text{and}\quad
 |\nabla_{(z_i,z_j)}\eta_\varepsilon(r)|^2
 =2|\nabla\eta_\varepsilon(r)|^2,
\]
discarding all repulsive factors except that of the pair \((i,j)\), and
integrating the Gaussian variables \(c\) and \(z_k\), \(k\ne i,j\), gives
\begin{equation}\label{eq:collision-cutoff}
 \int_{(\R^d)^N}|\nabla\eta_\varepsilon(z_i-z_j)|^2
 \ee^{-H(z)}\dd z
 \le C\varepsilon^{d-2}
       \exp[-c a\varepsilon^{2-d}]\longrightarrow0.
\end{equation}
Indeed, the remaining \(r\)-integral is bounded by
\[
 C\varepsilon^{-2}
 \int_{\varepsilon<|r|<2\varepsilon}
 \ee^{-a|r|^{2-d}}\dd r
 \le C\varepsilon^{d-2}\ee^{-ca\varepsilon^{2-d}}.
\]
For a compactly supported smooth \(v\), multiply \(v\) by the product of
these cutoffs over the finitely many pairs.  The \(L^2\) error and the term
where the gradient hits \(v\) vanish by dominated convergence, while the
terms where it hits a cutoff vanish by \eqref{eq:collision-cutoff}.  Ordinary
cutoffs at infinity then show that \(C_c^\infty(\Omega_N)\) is a core for
the closed weighted form.

\emph{Step 2: the many-particle ground-state transform.}
Let \(\Psi=\mathcal Z_{N,a}^{-1/2}\ee^{-H/2}\), set
\(A=\frac14|\nabla H|^2\), and put \(g=f\Psi\).  On
\(\Omega_N\), each pair potential is harmonic in both particle variables,
so
\begin{equation*}
 \Delta H=2\be dN.
\end{equation*}
Indeed, each of the \(N\) quadratic terms contributes \(2\be d\), while
Lemma~\ref{lem:superharmonic} gives zero for every pair term away from its
collision stratum.  On \(C_c^\infty(\Omega_N)\), expansion of the square and
integration by parts give
\begin{align*}
 \int_{(\R^d)^N}|\nabla f(z)|^2\dd\mu_{N,a}(z)
 &=\int_{\Omega_N}\left|\nabla g(z)+\tfrac12g(z)\nabla H(z)\right|^2\dd z\\
 &=\int_{\Omega_N}\left\{|\nabla g(z)|^2+
   \left(\tfrac14|\nabla H(z)|^2-\tfrac12\Delta H(z)\right)g(z)^2
   \right\}\dd z.
\end{align*}
Thus, after closure,
\begin{equation}
 \label{eq:many-transform}
 \int_{(\R^d)^N}|\nabla f(z)|^2\dd\mu_{N,a}(z)
 =\int_{\Omega_N}\{|\nabla g(z)|^2+(A(z)-\be dN)g(z)^2\}\dd z.
\end{equation}
The distributional Laplacian of \(W\) is supported on the collision strata,
but it contributes no term to \eqref{eq:many-transform}: the identity is
first proved on \(C_c^\infty(\Omega_N)\), and then closed using the
vanishing-cost collision cutoffs \eqref{eq:collision-cutoff}.

\emph{Step 3: properness of the transformed potential.}
We claim that \(A\) is proper on \(\Omega_N\): if a sequence leaves every
compact subset of \(\Omega_N\), then \(A\to\infty\).  Since
\(A=|\nabla H|^2/4\), it is enough to prove that every sequence with bounded
\(|\nabla H|\) remains in a compact subset of \(\Omega_N\), namely that all
positions stay bounded and all mutual distances stay bounded away from zero.
For later use, differentiation gives
\[
 \nabla_iH(z)
 =2\be z_i-a(d-2)\sum_{j\ne i}
   \frac{z_i-z_j}{|z_i-z_j|^d}.
\]

It suffices to show that a sequence \(z^{(n)}\) with bounded
\(|\nabla H(z^{(n)})|\) has bounded particle positions and a positive lower
bound on all pair distances.  Suppose first that
\(R_n=\max_i|z_i^{(n)}|\to\infty\).  After taking a subsequence, all
\(q_i^{(n)}=z_i^{(n)}/R_n\) converge.  Partition the indices into clusters
on which the limits \(q_i\) agree.  If \(C\) is one such cluster, the
pair forces internal to \(C\) cancel in pairs when summed over \(i\in C\):
the contribution of \((i,j)\) to \(\nabla_iH\) is the negative of its
contribution to \(\nabla_jH\).  For two distinct clusters \(C\ne D\), and
for \(i\in C\), \(j\in D\),
\begin{equation*}
 |z_i^{(n)}-z_j^{(n)}|
 =R_n(|q_C-q_D|+o(1)).
\end{equation*}
Since \(|q_C-q_D|>0\), every inter-cluster force has magnitude
\(O(R_n^{1-d})\).  There are only finitely many pairs, so summing the
gradient identity over \(i\in C\) and dividing by \(R_n\) gives
\begin{equation*}
 \frac1{R_n}\sum_{i\in C}\nabla_iH(z^{(n)})
 =2\be\sum_{i\in C}\frac{z_i^{(n)}}{R_n}+O(R_n^{-d})
 =2\be|C|q_C+o(1).
\end{equation*}
The left side tends to zero, so every cluster limit \(q_C\) is zero.  This
contradicts \(\max_i|q_i^{(n)}|=1\).  Thus the positions are bounded.

Suppose next that
\(r_n=\min_{i<j}|z_i^{(n)}-z_j^{(n)}|\to0\).  Pass to a subsequence on
which, for every pair, the ratio
\(|z_i^{(n)}-z_j^{(n)}|/r_n\) either stays bounded or tends to infinity.
The bounded-ratio relation partitions the indices.  Choose a nontrivial
cluster \(C\) containing a pair at distance \(r_n\), let \(\bar z_C^{(n)}\)
be its barycenter, and set
\begin{equation*}
 \xi_i^{(n)}=\frac{z_i^{(n)}-\bar z_C^{(n)}}{r_n},
 \qquad i\in C.
\end{equation*}
The relation is an equivalence relation: symmetry is clear, and transitivity
follows from the triangle inequality.  Thus all scaled distances within
\(C\) are bounded, so the barycenter condition
\(\sum_{i\in C}\xi_i^{(n)}=0\) makes every \(\xi_i^{(n)}\) bounded.  After
extraction these vectors converge to distinct \(\xi_i\), at least one pair
of which is at distance one.  Indeed, the definition of \(r_n\)
gives
\begin{equation*}
 |\xi_i^{(n)}-\xi_j^{(n)}|
 =\frac{|z_i^{(n)}-z_j^{(n)}|}{r_n}\ge1
 \qquad(i\ne j),
\end{equation*}
so no collision occurs in the limiting cluster.  The cluster contains a
pair at distance one, hence \(W_C(\xi)\) is finite and strictly positive.
Multiply the displayed formula for \(\nabla_iH\) by \(r_n^{d-1}\).  Since
the particle positions are bounded, the quadratic term tends to zero.  If
\(j\notin C\), then
\[
 r_n^{d-1}|z_i^{(n)}-z_j^{(n)}|^{1-d}
 =\left(\frac{r_n}{|z_i^{(n)}-z_j^{(n)}|}\right)^{d-1}
 \longrightarrow0.
\]
For \(i,j\in C\), homogeneity gives
\[
 r_n^{d-1}\nabla_{z_i}|z_i^{(n)}-z_j^{(n)}|^{2-d}
 =\nabla_{\xi_i}|\xi_i^{(n)}-\xi_j^{(n)}|^{2-d}.
\]
The quantity \(r_n^{d-1}\nabla_iH(z^{(n)})\) also tends to zero because
\(\nabla H(z^{(n)})\) is bounded.  Passing to the limit and dividing by
\(a>0\) therefore gives
\begin{equation*}
 \nabla_{\xi_i}W_C(\xi)=0\quad(i\in C),
 \qquad
 W_C(\xi)=\sum_{\substack{i<j\\i,j\in C}}|\xi_i-\xi_j|^{2-d}.
\end{equation*}
Differentiating \(W_C(t\xi)=t^{2-d}W_C(\xi)\) at \(t=1\) gives
Euler's identity
\begin{equation*}
 \sum_{i\in C}\xi_i\cdot\nabla_{\xi_i}W_C(\xi)
 =-(d-2)W_C(\xi)<0.
\end{equation*}
This is impossible if all \(\nabla_{\xi_i}W_C(\xi)\) vanish, because the
left side of Euler's identity would be zero while the right side is strictly
negative.  The assumed collision sequence therefore cannot exist, and the
properness claim follows.

\emph{Step 4: compact resolvent and simplicity of the ground state.}
Now add \(\be dN\|g\|_2^2\) to \eqref{eq:many-transform}.  The resulting
positive form norm is
\begin{equation*}
 \int_{\Omega_N}(|\nabla g(z)|^2+A(z)g(z)^2)\dd z.
\end{equation*}
By \eqref{eq:many-transform}, this quantity equals
\[
 \int_{(\R^d)^N}|\nabla f(z)|^2\dd\mu_{N,a}(z)
 +\be dN\|g\|_2^2,
\]
so its unit ball is bounded both in \(L^2\) and in the displayed positive
form norm.  For every member of that unit ball,
\[
 \int_{\{z\in\Omega_N:A(z)>M\}}g(z)^2\dd z
 \le\frac1M\int_{\Omega_N}A(z)g(z)^2\dd z\le\frac1M.
\]
Because \(A\) is proper, \(\{A\le M\}\) is compactly contained in
\(\Omega_N\).  Choose a bounded smooth open set
\(U_M\) with $\{A\le M\}\Subset U_M\Subset\Omega_N.$
The form-unit ball is bounded in \(H^1(U_M)\), so every sequence in it has
a subsequence converging in \(L^2(U_M)\).  Since
\(U_M^c\subset\{A>M\}\), the squared \(L^2\)-norm outside \(U_M\) is at
most \(M^{-1}\).  First choose \(M\) large and then use the Rellich
subsequence on \(U_M\); a diagonal argument produces a subsequence converging
in \(L^2(\Omega_N)\).  Thus the transformed form domain embeds compactly
into \(L^2\).  The unitary ground-state transform identifies the shifted
original form with this transformed positive form, so the shifted generator
has compact resolvent; shifting by the constant \(\be dN\) does not affect
compactness of the resolvent.  Therefore the original closed weighted form
also has compact resolvent.

Finally, \(\Omega_N\) is path connected: in dimension \(d\ge2\), particles
can be moved one at a time along paths avoiding the finitely many other
positions.  If the original weighted form of \(f\) is zero, then
\(\nabla f=0\) almost everywhere on \(\Omega_N\).  By the local Sobolev
regularity of the form domain and the usual chain-of-balls argument on a
connected open set, such an \(f\) is almost everywhere equal to a constant on
\(\Omega_N\).  Hence zero is a simple eigenvalue.  Compact resolvent makes
the spectrum discrete with no finite accumulation point, so the next
eigenvalue: the Poincar\'e gap, is strictly positive.
\end{proof}

\subsection{Completion of the proof}

\begin{proof}[Proof of Theorem~\ref{thm:main}]
Setting  $b_N:=\frac{\Theta_N}{N-1}
 =\frac{\chi}{N}\left(\frac{\be_N}{N}\right)^{d/2}$ and 
under the change of variables \(y_i=\sqrt{\be_N/N}\,x_i\), the law 
\eqref{eq:original-law} becomes
\[
 \mu_N(\dd y)
 \propto
 \exp\!\left[-\sum_{k=1}^N|y_k|^2
 -b_N\sum_{1\le k<\ell\le N}|y_k-y_\ell|^{2-d}\right]\dd y.
\]
Fix \(i\).  The conditional law of \(y_i\), given the other coordinates, is
\[
 \mu_i(\dd u\mid y_{\ne i})
 =\frac{\exp\!\left[-|u|^2-b_N\sum_{j\ne i}|u-y_j|^{2-d}\right]}
 {\displaystyle\int_{\R^d}
  \exp\!\left[-|v|^2-b_N\sum_{j\ne i}|v-y_j|^{2-d}\right]\dd v}\dd u.
\]
This formula defines a probability for every boundary configuration,
including coincident poles.  Its total Coulomb charge is exactly
\begin{equation}\label{eq:thetaN-total-charge}
 \Gamma_N=(N-1)b_N=\Theta_N.
\end{equation}

Let \(\eta_{d,1}>0\) be the small-charge threshold in
Theorem~\ref{thm:onesite}, and let \(M_{d,1}\) be the constant in
Proposition~\ref{prop:moving-center}.  Choose
\begin{equation}\label{eq:epsilon-d-choice}
 0<\varepsilon_d<
 \min\left\{\eta_{d,1},\frac{1}{4M_{d,1}}\right\}.
\end{equation}
By the assumption \eqref{eq:thetaN-assumption}, there are numbers
\(\widehat\Theta<\varepsilon_d\) and \(N_0\ge2\) such that
\(\Theta_N\le\widehat\Theta\) for every \(N\ge N_0\). 

Now we have prepared the ground to apply the Theorem
\ref{thm:onesite} that yielded a Poincar\'e Inequality for the 1-site measure: We apply it with Gaussian parameter \(\be=1\), which gives the
conditional Poincar\'e constant \(K_{d,1}\), uniformly in the boundary, for
all \(N\ge N_0\), i.e.
\[
 \Var_{\mu_i(\cdot\mid y_{\ne i})}(h) \leq
 K_{d,1}
 \int_{\R^d}|\nabla h(u)|^2\,
 \mu_i(\dd u\mid y_{\ne i}),
 \qquad h\in C_c^\infty(\R^d),
\]
uniformly in the site \(i\), the boundary configuration \(y_{\ne i}\),
and \(N\ge N_0\).

Moreover,  since
\(\widehat\Theta<(4M_{d,1})^{-1}<(2M_{d,1})^{-1}\), the hypothesis of
Proposition~\ref{prop:moving-center} is also satisfied. This allows us to estimate the spectral radius of the Dobrushin matrix. Indeed,  
if two boundary configurations differ only at \(y_j\), Proposition
\ref{prop:moving-center} and identity \eqref{eq: Wasserstein_TV} give for $i \neq j$ that $c_{ij}\le4M_{d,1}b_N.$
Hence the spectral radius of the Dobrushin matrix satisfies
\begin{equation}\label{eq:thetaN-dobrushin}
 \rho(C_N)\le\|C_N\|_\infty=:\kappa_N
 \le4M_{d,1}(N-1)b_N
 =4M_{d,1}\Theta_N
 \le\widehat\kappa:=4M_{d,1}\widehat\Theta<1.
\end{equation}

The conditional Poincar\'e inequality established above and
\eqref{eq:thetaN-dobrushin} verify the two hypotheses of
Proposition~\ref{prop:wu-criterion}. Therefore, for \(F\in C_c^\infty((\R^d)^N)\),

%In order to apply Proposition \ref{prop:Dobrushin-Wu}, we need to additionally verify the conditional one-site inequalities \eqref{eq:wu-conditional-PI} which are stated for smooth one-site test functions and then closed in Theorem~\ref{thm:onesite}; hence the hypotheses of Proposition~\ref{prop:wu-criterion} are satisfied for the smooth cylinder core and the resulting inequality extends to the closed global form. Proposition~\ref{prop:wu-criterion} now yields
\begin{equation}\label{eq:large-N-uniform-PI}
 \Var_{\mu_N}(F)
 \le \frac{K_{d,1}}{1-\widehat\kappa}
 \int_{(\R^d)^N}|\nabla F(y)|^2\dd\mu_N(y),
 \qquad N\ge N_0.
\end{equation}
By closure, this inequality extends to the domain of the closed global
Dirichlet form. This gives us a Poincar\'e inequality for all large $N$. 

Now for the remaining finitely many \(2\le N<N_0\), Lemma~\ref{lem:fixed-N-gap}, applied with Gaussian
parameter \(\be=1\) and interaction coefficient \(a=b_N\), gives a finite
Poincar\'e constant \(K_N\).  Therefore
\[
 K_*:=\max\left\{
 \frac{K_{d,1}}{1-\widehat\kappa},
 \max_{2\le N<N_0}K_N\right\}<\infty
\]
is a Poincar\'e constant for every \(\mu_N\).

For $F(y)=f\!\left(\sqrt{\frac{N}{\be_N}}\,y\right),$
one has \(\nabla_{y_i}F=\sqrt{N/\be_N}\,\nabla_{x_i}f\), and hence
\[
 \Var_{\mu_N}(F)=\Var_{P_{N,\be_N}^{d}}(f),
 \qquad
 \int|\nabla_yF|^2\dd\mu_N
 =\frac{N}{\be_N}\int|\nabla_xf|^2\dd P_{N,\be_N}^{d}.
\]
Consequently,
\[
 \Var_{P_{N,\be_N}^{d}}(f)
 \le K_*\frac{N}{\be_N}\int|\nabla_xf|^2\dd P_{N,\be_N}^{d}
 =\frac{K_*N}{\al_N}\cE_N(f,f),
\]
which proves the lower bound $\lambda_N^d(\be_N)\ge\frac{\al_N}{K_*N}.$

For the quantitative assertion, suppose
\(\sup_N\Theta_N\le\bar\Theta<\varepsilon_d\).  Then the preceding
argument applies with \(N_0=2\),
\(\widehat\Theta=\bar\Theta\), and no exceptional finite set.  In this case one may take
\[
 K_*=\frac{K_{d,1}}{1-4M_{d,1}\bar\Theta},
 \qquad
 c_{d,\bar\Theta}
 =\frac{1-4M_{d,1}\bar\Theta}{K_{d,1}}.
\]

For the upper bound, let the center of mass $c=N^{-1/2}\sum_{i=1}^N x_i.$
Writing \(r_i=x_i-c/\sqrt N\), the Gibbs density factorizes in the
\(c\)-coordinate with factor \(\exp[-\be_N|c|^2/N]\), since due to the translation invariance of the interactions, these do not depend on \(c\) at all.  Consequently each component
\(c_k\) has variance \(N/(2\be_N)\).  Since \(|\nabla c_k|^2=1\),
\[
 \cE_N(c_k,c_k)=\frac{\al_N}{\be_N},
\]
and its Rayleigh quotient is \(2\al_N/N\).  This proves the upper bound.
\end{proof}

\begin{remark}[Constants and uniformity]
The argument is nonquantitative: it proves existence of
\(\eta_{d,\be}\) and \(K_{d,\be}\) without useful numerical values.
Their uniformity over finite pole configurations, including collisions, is
exactly what is needed for the many-particle conditional measures.
\end{remark}

\begin{appendix}
    
\section{The logarithmic Coulomb gas in dimension two}
\label{app:planar-log}

We explain here how the proof of Theorem~\ref{thm:main} adapts to
\(K(x)=-\log|x|\).  Its structure is unchanged: the one-site estimate comes
from harmonicity away from the poles and the ground-state transform, while
the many-particle estimate follows from the moving-pole bound and
Proposition~\ref{prop:wu-criterion}.  Only the estimates specific to dimension
two are given below. The only real subtlety is the change of sign that the $\log$ exhibits and the different scaling behaviour. None of this requires a different method, but just some care.

Fix \(\chi\ge0\) and positive sequences \((\be_N)_{N\ge2}\) and
\((\al_N)_{N\ge2}\).  For \(N\ge2\), let
\begin{equation}\label{eq:planar-law}
 P_{N,\be_N}^{2}(\dd x)
 =\frac1{Z_{N,\be_N}^{2}}
 \exp\!\left[-\frac{\be_N}{N}\sum_{i=1}^N|x_i|^2\right]
 \prod_{1\le i<j\le N}|x_i-x_j|^{\be_N\chi/N^2}\dd x,
\end{equation}
with Dirichlet form
\[
 \cE_N(f,f)=\frac{\al_N}{\be_N}
 \int_{(\R^2)^N}|\nabla f|^2\dd P_{N,\be_N}^{2}.
\]
Write \(\lambda_N^2(\be_N)\) for the associated spectral gap and set
\begin{equation}\label{eq:planar-theta}
 \Theta_N^{(2)}
 :=\frac{N-1}{N}\chi\frac{\be_N}{N}
 =\frac{(N-1)\be_N\chi}{N^2}.
\end{equation}
The result is then, as for $d\ge 3:$

\begin{theorem}\label{thm:planar-log}
There exists \(\varepsilon_2>0\) such that, if
\[
 \limsup_{N\to\infty}\Theta_N^{(2)}<\varepsilon_2,
\]
then there is \(c_{2,\chi,(\be_N)}>0\), independent of \(N\), for which
\begin{equation}\label{eq:planar-gap}
 c_{2,\chi,(\be_N)}\frac{\al_N}{N}
 \le \lambda_N^2(\be_N)
 \le 2\frac{\al_N}{N},
 \qquad N\ge2.
\end{equation}
\end{theorem}

For fixed \(\be_N\equiv\be>0\), the hypothesis is automatic since
\(\Theta_N^{(2)}=O(N^{-1})\).  If
\(\be_N/N\to\beta^{\mathrm{ph}}\), it reduces to
\(\chi\beta^{\mathrm{ph}}<\varepsilon_2\).
Let us now point out the essential differences to the $d\ge 3 $ case:
\subsection*{Proof sketch of Theorem~\ref{thm:planar-log}}

\subsection{Rescaling and conditional measures.}
As in the proof of Theorem~\ref{thm:main}, set
\(y_i=\sqrt{\be_N/N}\,x_i\).  The additive constant (due to $\sqrt{\be_N/N}$) coming from the logarithm
is absorbed into the partition function, and
\eqref{eq:planar-law} becomes
\begin{equation}\label{eq:planar-scaled-law}
 \mu_N(\dd y)
 \propto
 \exp\!\left[-\sum_{i=1}^N|y_i|^2\right]
 \prod_{i<j}|y_i-y_j|^{a_N}\dd y,
 \qquad a_N:=\frac{\be_N\chi}{N^2}.
\end{equation}
More precisely, if \(\mathcal Z_N\) is the normalizing constant in
\eqref{eq:planar-scaled-law}, then
\[
 Z_{N,\be_N}^{2}
 =\left(\frac{N}{\be_N}\right)^{
   N+\frac{a_N}{2}\binom N2}\mathcal Z_N.
\]
This factor plays no role in the spectral-gap estimate.  The conditional law
at site \(i\) is
\begin{equation}\label{eq:planar-conditional}
 \mu_i(\dd u\mid y_{\ne i})
 =\frac{e^{-|u|^2}\prod_{j\ne i}|u-y_j|^{a_N}}
 {\displaystyle\int_{\R^2}e^{-|v|^2}
   \prod_{j\ne i}|v-y_j|^{a_N}\dd v}\dd u.
\end{equation}
The total logarithmic charge in this conditional measure is
\((N-1)a_N=\Theta_N^{(2)}\), the planar counterpart of
\eqref{eq:thetaN-total-charge}.

As in the $d\ge 3$ case, we first try to get a spectral gap estimate for a single particle. 

\subsection{The one-site estimate.}
Much of the difficulty of the $d=2$ potential is related to the fact that it changes sign. This needs an additional "pre-processing step".
Let
\[
 \rho=\sum_{\ell=1}^M\gamma_\ell\delta_{y_\ell},
 \quad M<\infty,\quad \gamma_\ell>0,
 \quad
 \Gamma(\rho)=\sum_{\ell=1}^M\gamma_\ell,
 \quad
 U_\rho^{\log}(x)=-\sum_{\ell=1}^M\gamma_\ell\log|x-y_\ell|.
\]
On \(\R^2\), let
\[
 \gamma_\be(\dd x)=Z_{2,\be}^{-1}e^{-\be|x|^2}\dd x,
 \qquad
 Z_{2,\be}=\int_{\R^2}e^{-\be|x|^2}\dd x.
\]
The corresponding one-site measure is
\begin{equation}\label{eq:planar-one-site-law}
 \nu_\rho^{\log}(\dd x)
 =\frac1{Z_\rho^{\log}}
 e^{-\be|x|^2-U_\rho^{\log}(x)}\dd x
 =\frac1{Z_\rho^{\log}}e^{-\be|x|^2}
 \prod_{\ell=1}^M|x-y_\ell|^{\gamma_\ell}\dd x.
\end{equation}
Since \(-\log|x|\) is harmonic away from the origin,
\(\Delta U_\rho^{\log}=0\) off the poles, as in
Lemma~\ref{lem:superharmonic}.  Distributionally,
\(-\Delta U_\rho^{\log}=2\pi\rho\).

The first essential difference is that, unlike the Coulomb weight in dimension \(d\ge3\),
\(e^{-U_\rho^{\log}}\) need not be bounded by one.  We therefore remove the
constant carried by a distant pole and write
\[
 \langle y\rangle=(1+|y|^2)^{1/2},
 \qquad
 \ell_y(x)=\log\frac{|x-y|}{\langle y\rangle},
 \qquad
 L_\rho(x)=\sum_\ell\gamma_\ell\ell_{y_\ell}(x).
\]
For every \(0<s<2\),
\begin{equation}\label{eq:planar-gaussian-log-moment}
 \sup_{y\in\R^2}
 \int_{\R^2}e^{s|\ell_y(x)|}\dd\gamma_\be(x)<\infty.
\end{equation}
For the positive part, use
\(|x-y|/\langle y\rangle\le1+|x|\).  For the negative part, separate bounded
\(y\) from \(|y|>2\).  The singularity is integrable because
\(\int_0^1r^{1-s}\dd r<\infty\), and the Gaussian controls a neighborhood of
a distant center.

Choose \(0<\Gamma_0<1\) sufficiently small.  Constants in the rest of the
appendix may depend on \(\Gamma_0\).  Uniformly for
\(\Gamma(\rho)\le\Gamma_0\),
\begin{equation}\label{eq:planar-log-normalization}
 Z_\rho^{\log}
 =Z_{2,\be}\prod_\ell\langle y_\ell\rangle^{\gamma_\ell}m_\rho,
 \qquad
 m_\rho:=\int e^{L_\rho}\dd\gamma_\be=1+O(\Gamma(\rho)),
\end{equation}
and
\begin{equation}\label{eq:planar-gaussian-TV}
 \|\nu_\rho^{\log}-\gamma_\be\|_{\TV}
 \le C_\be\Gamma(\rho).
\end{equation}
Jensen's inequality and the moment bound above give
\[
 \|L_\rho\|_{L^2(\gamma_\be)}+
 \|L_\rho\|_{L^4(\gamma_\be)}
 \le C_\be\Gamma(\rho),
 \qquad
 \int e^{2|L_\rho|}\dd\gamma_\be\le C_\be.
\]
Together with \(|e^t-1|\le|t|e^{|t|}\), these estimates prove
\eqref{eq:planar-log-normalization}--\eqref{eq:planar-gaussian-TV}.  The bound
\(|e^{t/2}-1|^2\le C t^2e^{|t|}\) and H\"older's inequality similarly give
\begin{equation}\label{eq:planar-ground-state-convergence}
 \|e^{L_\rho/2}-1\|_{L^2(\gamma_\be)}
 \le C_\be\Gamma(\rho).
\end{equation}
The exponential moment remains uniform after multiplication by one additional
logarithmic factor, which gives
\begin{equation}\label{eq:planar-centered-log}
 \sup_{y\in\R^2}
 \int_{\R^2}
 \left|\log|x-y|-\nu_\rho^{\log}(\log|\cdot-y|)\right|
 \dd\nu_\rho^{\log}(x)
 \le C_\be,
\end{equation}
uniformly for \(\Gamma(\rho)\le\Gamma_0\), even when poles collide or tend to
infinity.

The field estimate used in Lemma~\ref{lem:weak-field} carries over, since the gradient of the $\log$ is of the same form as our $d\ge 3$ computation, and becomes
\begin{equation}\label{eq:planar-weak-field}
 \left|\left\{x\in\R^2:
 |\nabla U_\rho^{\log}(x)|>t\right\}\right|
 \le C\left(\frac{\Gamma(\rho)}t\right)^2,
 \qquad t>0.
\end{equation}

The next result follows then our $d\ge 3$ strategy as well and one finds for Theorem \ref{thm:onesite}

\begin{proposition}\label{prop:planar-one-site}
For every \(\be>0\), there are \(\eta_{2,\be}>0\) and
\(K_{2,\be}<\infty\) such that, whenever
\(\Gamma(\rho)\le\eta_{2,\be}\),
\begin{equation}\label{eq:planar-one-site-input}
 \Var_{\nu_\rho^{\log}}(f)
 \le K_{2,\be}
 \int_{\R^2}|\nabla f|^2\dd\nu_\rho^{\log}.
\end{equation}
\end{proposition}

\begin{proof}
We first combine poles at the same point.  If \(\bar\gamma>0\) is the total
charge at one point, the annular cutoff from
Lemma~\ref{lem:singular-transform} has weighted cost
\begin{equation}\label{eq:planar-pole-cutoff}
 O\!\left(\varepsilon^{-2}
 \int_\varepsilon^{2\varepsilon}r^{\bar\gamma}r\dd r\right)
 =O(\varepsilon^{\bar\gamma})\longrightarrow0.
\end{equation}
It follows, as in our $d\ge 3$ argument, that smooth functions supported away from the poles form a core.
Set
\[
 H_\rho=\be|x|^2+U_\rho^{\log}(x),
 \qquad
 \psi_\rho=(Z_\rho^{\log})^{-1/2}e^{-H_\rho/2},
 \qquad
 A_\rho=\frac14|2\be x+\nabla U_\rho^{\log}|^2.
\]
Then \eqref{eq:planar-log-normalization} gives
\[
 \psi_\rho
 =\psi_0\frac{e^{L_\rho/2}}{\sqrt{m_\rho}},
 \qquad
 \psi_0=Z_{2,\be}^{-1/2}e^{-\be|x|^2/2}.
\]
By \eqref{eq:planar-ground-state-convergence},
\(\psi_\rho\to\psi_0\) in \(L^2\) as \(\Gamma(\rho)\to0\), uniformly in
the pole configuration.  The ground-state calculation in
Lemma~\ref{lem:singular-transform}, first off the poles and then by closure,
gives
\begin{equation}\label{eq:planar-ground-state-form}
 \int|\nabla f|^2\dd\nu_\rho^{\log}
 =Q_\rho(g)
 :=\int_{\R^2}\left(|\nabla g|^2+(A_\rho-2\be)g^2\right)\dd x,
 \qquad g=f\psi_\rho.
\end{equation}

It remains to check the compactness step in dimension two.  We proceed as in the Coulomb case $d \ge 3$: Suppose that the
proposition is false.  As in the proof of Theorem~\ref{thm:onesite}, we may
then find \(\Gamma(\rho_n)\to0\) and centered functions \(f_n\), normalized by
\(\nu_{\rho_n}^{\log}(f_n^2)=1\), whose energies tend to zero.  Set
\(g_n=f_n\psi_{\rho_n}\).  Then
\(\|g_n\|_2=1\), \(Q_{\rho_n}(g_n)\to0\), and
\(g_n\perp\psi_{\rho_n}\).  Then
\eqref{eq:planar-ground-state-convergence} gives
\(\psi_{\rho_n}\to\psi_0\) in \(L^2\).

Since \(A_{\rho_n}\ge0\), \eqref{eq:planar-ground-state-form} gives uniform
bounds on \((g_n)\) in \(H^1(\R^2)\) and on
\(\int A_{\rho_n}g_n^2\).  At this point the critical Sobolev estimate used
for \(d\ge3\) is replaced by the two-dimensional inequality
\(\|g_n\|_4^2\le C\|g_n\|_2\|\nabla g_n\|_2\).  Combining it with
\eqref{eq:planar-weak-field} gives, for fixed \(R>0\),
\begin{equation}\label{eq:planar-transformed-tail}
 \int_{|x|>R}g_n^2\dd x
 \le C_\be\frac{\Gamma(\rho_n)}R+
 \frac{C_\be}{R^2}.
\end{equation}
Indeed, let
\[
 E_{n,R}=\{x:|x|>R,\ |\nabla U_{\rho_n}^{\log}(x)|>\be|x|\}.
\]
Then~\eqref{eq:planar-weak-field} gives
\(|E_{n,R}|\le C_\be(\Gamma(\rho_n)/R)^2\), and H\"older's inequality
controls \(\int_{E_{n,R}}g_n^2\) by \(C_\be\Gamma(\rho_n)/R\).  On the
complement of \(E_{n,R}\), one has
\(A_{\rho_n}\ge\be^2|x|^2/4\), which gives the second term in
\eqref{eq:planar-transformed-tail}.  Rellich compactness on balls and
\eqref{eq:planar-transformed-tail} now give, along a subsequence,
\(g_n\to g\) strongly in \(L^2(\R^2)\), with \(\|g\|_2=1\) and
\(g\perp\psi_0\).

On a fixed ball, away from
\(\{|\nabla U_{\rho_n}^{\log}|>\delta\}\),
\[
 A_{\rho_n}(x)
 \ge\be^2|x|^2-\be R\delta.
\]
The weak-field bound and the \(L^4\) estimate show that this exceptional set
carries vanishing \(g_n^2\)-mass.  Letting first \(n\to\infty\), then
\(\delta\downarrow0\), and finally \(R\uparrow\infty\), gives
\[
 Q_0(g)\le\liminf_{n\to\infty}Q_{\rho_n}(g_n)=0,
 \qquad
 Q_0(g)=\int_{\R^2}|\nabla g+\be xg|^2\dd x.
\]
Since \(\ker Q_0=\operatorname{span}\{\psi_0\}\) and
\(g\perp\psi_0\), this forces \(g=0\), contrary to \(\|g\|_2=1\).
\end{proof}
The replacement of Proposition \ref{prop:moving-center} is then
\begin{proposition}
\label{prop:planar-moving-center}
After decreasing \(\Gamma_0\), if necessary, there is
\(C_{2,\be}<\infty\) such that, for every \(a\ge0\), every \(\rho\) with
\(\Gamma(\rho)+a\le\Gamma_0\), and every \(y,y'\in\R^2\),
\begin{equation}\label{eq:planar-moving-input}
 \left\|
 \frac{|\,\cdot-y|^a\nu_\rho^{\log}}
      {\nu_\rho^{\log}(|\,\cdot-y|^a)}
 -
 \frac{|\,\cdot-y'|^a\nu_\rho^{\log}}
      {\nu_\rho^{\log}(|\,\cdot-y'|^a)}
 \right\|_{\TV}
 \le C_{2,\be}a.
\end{equation}
\end{proposition}

\begin{proof}
Here, we just give the main intermediate steps with the full proof following closely its $d\ge 3$ counterpart:
Let \(M_{2,\be}^{\log}\) denote the uniform constant in
\eqref{eq:planar-centered-log}.  Fix \(y\) and consider
\[
 \dd\nu_t
 =\frac{e^{ta\ell_y}}{\nu_\rho^{\log}(e^{ta\ell_y})}
 \dd\nu_\rho^{\log},
 \qquad 0\le t\le1.
\]
The factor \(\langle y\rangle^{ta}\) cancels in the normalization, so this is
the normalized law obtained by adding a pole of charge \(ta\) at \(y\).
Furthermore,
\(\nu_\rho^{\log}(e^{ta\ell_y})
=m_{\rho+ta\delta_y}/m_\rho\), so
\eqref{eq:planar-log-normalization} controls the normalizing factor uniformly
along the path.  The exponential-log estimate
\eqref{eq:planar-gaussian-log-moment} permits differentiation; one may first
truncate the logarithm and then remove the truncation.  For \(|h|\le1\),
\eqref{eq:planar-centered-log} gives
\[
 \frac{\dd}{\dd t}\nu_t(h)
 =a\operatorname{Cov}_{\nu_t}(h,\log|\cdot-y|),
 \qquad
 \left|\frac{\dd}{\dd t}\nu_t(h)\right|
 \le M_{2,\be}^{\log}a.
\]
After integration in \(t\) and taking the supremum over \(|h|\le1\), we find
that adding a pole at \(y\) changes the law by at most
\(M_{2,\be}^{\log}a\) in the full total-variation norm.  Applying the same
estimate at \(y'\), with \(\nu_\rho^{\log}\) as the intermediate law, proves
the result with \(C_{2,\be}=2M_{2,\be}^{\log}\).
\end{proof}

\subsection{Dobrushin estimate.}
Fix \(i\ne j\) and compare two boundary configurations that agree away from
site \(j\).  In \eqref{eq:planar-conditional}, all factors indexed by
\(k\ne i,j\) are common to the two conditional laws; only the logarithmic
pole of charge \(a_N\) at \(y_j\) is moved.  The common factors define
\[
 \rho=a_N\sum_{k\ne i,j}\delta_{y_k},
 \qquad
 \Gamma(\rho)+a_N=(N-1)a_N=\Theta_N^{(2)}.
\]
Proposition~\ref{prop:planar-moving-center}, together with the convention
\eqref{eq:TV-convention} and the discrete metric in
\eqref{eq:wu-matrix}, therefore gives
\begin{equation}\label{eq:planar-dobrushin}
 c_{ij}\le \frac{C_{2,1}}{2}a_N,
 \qquad
 \rho(C_N)\le\|C_N\|_\infty
 \le \frac{C_{2,1}}{2}(N-1)a_N
 =\frac{C_{2,1}}{2}\Theta_N^{(2)}.
\end{equation}
Choose \(\varepsilon_2\) so that
\[
 0<\varepsilon_2<\min\{\eta_{2,1},\Gamma_0,2/C_{2,1}\}.
\]
Under the hypothesis of Theorem~\ref{thm:planar-log}, the one-site estimate
\eqref{eq:planar-one-site-input} holds uniformly and \(\rho(C_N)<1\) for all
sufficiently large \(N\).  Proposition~\ref{prop:wu-criterion} then yields a
Poincar\'e constant for \(\mu_N\) that is independent of \(N\).
If \(\sup_N\Theta_N^{(2)}\le\bar\Theta<\varepsilon_2\), the rescaled
constant may be taken as
\[
 \frac{K_{2,1}}{1-(C_{2,1}/2)\bar\Theta},
 \qquad
 c_{2,\bar\Theta}
 =\frac{1-(C_{2,1}/2)\bar\Theta}{K_{2,1}}.
\]
Under the limsup assumption, choose
\[
 \limsup_{N\to\infty}\Theta_N^{(2)}
 <\widehat\Theta<\varepsilon_2.
\]
The preceding estimate applies for all sufficiently large \(N\); the finitely
many remaining values are covered by the fixed-particle argument below.

\subsection{Fixed particle number work-around}
Similarly to our argument for $d\ge 3$, we need to cover the finitely many particle numbers where the Dobrushin estimate may not be effective,
the argument of Lemma~\ref{lem:fixed-N-gap} requires only a different cutoff near collisions.  Fix \(N\), \(\be>0\), and \(a>0\), and write
\[
 \mu_{N,a}(\dd z)
 \propto e^{-\be\sum_i|z_i|^2}
 \prod_{i<j}|z_i-z_j|^a\dd z,
 \qquad
 \Omega_N=\{z:z_i\ne z_j\text{ for }i\ne j\}.
\]
Near a pair diagonal an annular cutoff has weighted energy
\begin{equation}\label{eq:planar-collision-cutoff}
 O\!\left(\varepsilon^{-2}
 \int_\varepsilon^{2\varepsilon}r^{a+1}\dd r\right)
 =O(\varepsilon^a)\longrightarrow0.
\end{equation}
For a compactly supported smooth test function, the remaining pair factors
are bounded above on its support.  Repeating this cutoff construction over
the finitely many pair diagonals, and using the product rule and dominated
convergence, shows that \(C_c^\infty(\Omega_N)\) is a core.  Thus removing the
collision set leaves the closed weighted form unchanged.

To prove compactness, set
\[
 W_a(z)=\prod_{i<j}|z_i-z_j|^a,
 \qquad b=\nabla\log W_a,
 \qquad q=a\binom N2,
\]
and let
\(\Psi=Z_{N,a}^{-1/2}W_a^{1/2}e^{-\be|z|^2/2}\), \(\xi=f\Psi\).
The homogeneity identities
\[
 z\mathbin{\cdot}b=q,
 \qquad
 \operatorname{Div}b=0\quad\text{on }\Omega_N
\]
give, first on the collision-free core and then by closure,
\begin{equation}\label{eq:planar-many-transform}
 \int|\nabla f|^2\dd\mu_{N,a}
 =\int_{\Omega_N}\left(
 |\nabla\xi|^2+
 \left[\be^2|z|^2-\be(2N+q)+\frac14|b(z)|^2\right]\xi^2
 \right)\dd z.
\end{equation}
As in \eqref{eq:planar-ground-state-form}, this closed punctured-space
identity contains no additional distributional masses on the diagonals; the
realization is selected by \eqref{eq:planar-collision-cutoff}.  Dropping the
nonnegative term \(\frac14|b|^2\) gives
\[
 \int_{\Omega_N}(|\nabla\xi|^2+\be^2|z|^2\xi^2)\dd z
 \le
 \int|\nabla f|^2\dd\mu_{N,a}
 +\be(2N+q)\|f\|_{L^2(\mu_{N,a})}^2.
\]
Hence the transformed form domain embeds continuously into the full
harmonic-oscillator form domain on \(\R^{2N}\).  The latter embeds compactly
into \(L^2\), and unitarity of the ground-state transform gives compact
resolvent for the weighted form.

Radial cutoffs at infinity, combined with
\eqref{eq:planar-collision-cutoff}, show that the constant function belongs to
the form domain.  If a form-domain function has zero energy, then local
equivalence of the density and Lebesgue measure on compact subsets of
\(\Omega_N\) places it in \(H^1_{\mathrm{loc}}(\Omega_N)\) with zero weak
gradient.  Since \(\Omega_N\) is path connected in dimension two, the
function is constant.  Thus zero is a simple eigenvalue, and compact
resolvent yields a strictly positive next eigenvalue.  When \(a=0\), this is
the usual product-Gaussian gap.

\subsection{Return to the original variables.}
For \(F(y)=f(\sqrt{N/\be_N}\,y)\),
\[
 \int|\nabla_yF|^2\dd\mu_N
 =\frac{N}{\be_N}
 \int|\nabla_xf|^2\dd P_{N,\be_N}^{2}.
\]
The rescaled Poincar\'e inequality gives the lower bound in
\eqref{eq:planar-gap}.  Since the interaction in \eqref{eq:planar-law}
depends only on particle differences, the center-of-mass calculation in the
proof of Theorem~\ref{thm:main} gives the upper bound \(2\al_N/N\).  

\end{appendix}

\end{document}